\documentclass[12pt]{amsart}
\usepackage{amsmath,amsthm,amssymb,bm,bbm,dsfont}
\usepackage{amsaddr}
\usepackage{upgreek}
\usepackage[left=2cm,right=2cm,top=2cm,bottom=2cm]{geometry}
\usepackage{setspace}
\usepackage{graphicx}
\usepackage{verbatim}
\usepackage{float}
\restylefloat{table}
\usepackage{placeins}
\usepackage{array}
\usepackage{booktabs}
\usepackage{threeparttable}
\usepackage[update,prepend]{epstopdf}
\usepackage{multirow}
\usepackage{amsfonts,amssymb,dsfont}
\usepackage{mathrsfs}  
\usepackage[usenames,dvipsnames]{color}
\usepackage[hidelinks]{hyperref}
\hypersetup{
	unicode=false,          
	pdftoolbar=true,        
	pdfmenubar=true,        
	pdffitwindow=false,     
	pdfstartview={FitH},    
	pdftitle={My title},    
	pdfauthor={Author},     
	pdfsubject={Subject},   
	pdfcreator={Creator},   
	pdfproducer={Producer}, 
	pdfkeywords={keyword1} {key2} {key3}, 
	pdfnewwindow=true,      
	colorlinks=true,        
	linkcolor=Red,          
	citecolor=ForestGreen,  
	filecolor=Magenta,      
	urlcolor=BlueViolet,    
}
\usepackage{doi}
\usepackage{url}
\usepackage{caption, subcaption}
\usepackage{enumitem}

\makeatletter
\ifx\@NODS\undefined%

\let\mathbb=\mathds
\else%
\fi
\makeatother

\usepackage[backend=bibtex,
hyperref=true,
url=true,
isbn=true,
style=ieee,
sorting=none,
block=none]{biblatex}
\bibliography{reference.bib}

\newcommand{\be}{{\mathbf e}}

\def\0{{\mathbf{0}}}
\def\1{{\mathbf{1}}}
\def\2{{\mathbf{2}}}
\def\3{{\mathbf{3}}}
\def\4{{\mathbf{4}}}
\def\5{{\mathbf{5}}}
\def\6{{\mathbf{6}}}

\def\7{{\mathbf{7}}}
\def\8{{\mathbf{8}}}
\def\9{{\mathbf{9}}}

\def\be{\begin{equation}}
\def\ee{\end{equation}}
\def\bea{\begin{eqnarray}}
\def\eea{\end{eqnarray}}

\theoremstyle{plain}
\newtheorem{theo}{Theorem}[section]

\newtheorem{prop}{Proposition}[section]
\newtheorem{lemm}{Lemma}[section]
\newtheorem{coro}{Corollary}[section]

\theoremstyle{definition}
\newtheorem{defn}{Definition}[section]

\theoremstyle{remark}
\newtheorem{remark}{Remark}[section]

\numberwithin{equation}{section}

\newcommand\xqed[1]{%
	\leavevmode\unskip\penalty9999 \hbox{}\nobreak\hfill
	\quad\hbox{#1}}
\newcommand\Endremark{\xqed{$\Diamond$}}

\DeclareMathOperator{\Tr}{Tr}
\DeclareMathOperator{\tr}{\overline{tr}}

\begin{document}

\let\origmaketitle\maketitle
\def\maketitle{
	\begingroup
	\def\uppercasenonmath##1{} 
	\let\MakeUppercase\relax 
	\origmaketitle
	\endgroup
}

\title{{Characterisations of Matrix and Operator-Valued $\Phi$-Entropies, and Operator Efron-Stein Inequalities}}
\author{{Hao-Chung Cheng$^{1,2}$ and Min-Hsiu Hsieh$^2$}}
\address{\small  $^{1}$Graduate Institute Communication Engineering, National Taiwan University, Taiwan (R.O.C.) \\ $^{2}$Centre for Quantum Computation and Intelligent Systems, \\Faculty of Engineering and Information Technology, University of Technology Sydney, Australia} 
\email{\href{mailto:F99942118@ntu.edu.tw}{F99942118@ntu.edu.tw}$^1$}
\email{\href{mailto:Min-Hsiu.Hsieh@uts.edu.au}{Min-Hsiu.Hsieh@uts.edu.au}$^2$}

\begin{abstract}
We derive new characterisations of the matrix $\mathrm{\Phi}$-entropy functionals introduced in [\href{http://dx.doi.org/10.1214/ejp.v19-2964}{\textit{Electron.~J.~Probab., \textbf{19}(20): 1--30, 2014}}]. Notably, all known equivalent characterisations of the classical $\Phi$-entropies have their matrix correspondences. Next, we propose an operator-valued generalisation of the matrix  $\Phi$-entropy functionals, and prove their subadditivity under L\"owner partial ordering. Our results demonstrate that the subadditivity of operator-valued $\Phi$-entropies is equivalent to the convexity of various related functions. This result can be used to demonstrate an interesting result in quantum information theory: the matrix $\Phi$-entropy of a quantum ensemble is \emph{monotone} under unital quantum channels.
Finally, we derive the operator Efron-Stein inequality to bound the operator-valued variance of a random matrix.
\end{abstract}

\maketitle

\section{Introduction} 

The introduction of $\mathrm{\Phi}$-entropy functionals can be traced back to the early days of information theory \cite{Csi63,Csi67} and convex analysis \cite{AS66,BR82a,BR82b,BR82c}, where the notion of $\phi$-divergence is defined.  Formally, given a non-negative real random variable $Z$ and a smooth convex function $\Phi$,
the $\Phi$-entropy functional refers to
\[
H_\mathrm{\Phi}(Z) = \mathbb{E} \mathrm{\Phi}(Z) - \mathrm{\Phi}(\mathbb{E} Z). 
\]
By Jensen's inequality, it is not hard to see that the quantity $H_{\Phi}(Z)$ is non-negative. Hence, the $\Phi$-entropy functional can be used as an entropic measure to characterise the uncertainty of the random variable $Z$.

The investigation of general properties of classical $\mathrm{\Phi}$-entropies has enjoyed great success in physics, probability theory, information theory and computer science. Of these, the \textit{subadditivity} (or the \emph{tensorisation}) property \cite{Han78, BL97,Led97} has led to the derivations of the logarithmic Sobolev \cite{Gro75}, $\mathrm{\Phi}$-Sobolev \cite{LO00} and Poincar\'{e} inequalities \cite{ABC+00}, which in turn, is a crucial step toward the powerful \emph{entropy method} in concentration inequalities \cite{BBM+05,Mas07,BLM13}
and analysis of Markov semigroups \cite{BGL13}. 



Let $Z=f(X_1,\cdots,X_n)$ be a random variable defined on $n$ independent random variables $(X_1,\cdots,X_n)$. We say $H_\mathrm{\Phi}(Z)$ is \textit{subadditive} if 
\[
H_\mathrm{\Phi}(Z)\leq  \sum_{i=1}^n \mathbb{E} \left[
\mathbb{E}_i \mathrm{\Phi}(Z) -\mathrm{\Phi}(\mathbb{E}_iZ)
\right],
\]
where 
$\mathbb{E}_i$ denotes the conditional expectation with respect to $X_i$. 
L.~Gross first observed that the ordinary entropy functional $H_{u\log u}(Z)$ is subadditive in his seminal paper \cite{Gro75}. Later on, equivalent characterisations of the subadditive entropy class (see Theorem \ref{thm_classical}) are established \cite{LO00, Cha03, Cha06}, which prove to be useful in other contexts such as stochastic processes \cite{Cha03, Cha06}.

Parallel to the classical $\Phi$-entropies, Chen and Tropp \cite{CT14} introduced the notion of matrix $\mathrm{\Phi}$-entropy functionals. Namely, for a positive semi-definite random matrix $\bm{Z}$, the matrix
$\mathrm{\Phi}$-entropy functional is defined as
\[ \label{eq:tentropy}
H_\mathrm{\Phi}(\bm{Z})\triangleq \tr\left[\mathbb{E}\mathrm{\Phi}(\bm{Z})-\mathrm{\Phi}(\mathbb{E}\bm{Z})\right],
\]
where $\tr$ is the normalised trace.
The class of subadditive matrix $\Phi$-entropy functionals is  characterised in terms of the second derivative of their representing functions. Unlike its classical counterpart, only a few connections between the matrix $\mathrm{\Phi}$-entropy functionals and  other convex forms of the same functions have been established  \cite{HZ14, PV14} prior to this current work.  

In this paper, we establish equivalent characterisations of the matrix $\mathrm{\Phi}$-entropy functionals defined in \cite{CT14}. Our results show that matrix $\mathrm{\Phi}$-entropy functionals satisfy all known equivalent statements that classical $\mathrm{\Phi}$-entropy functions satisfy \cite{Cha03, Cha06, BLM13}. Our results provide additional justification to its original definition of the matrix $\mathrm{\Phi}$-entropy functionals (see Table \ref{table:trace}). The equivalences between matrix $\mathrm{\Phi}$-entropy functionals and other convex forms of the function $\mathrm{\Phi}$ advance our understanding of the class of entropy functions. Moreover, it allows to unify the study of matrix concentration inequalities and matrix $\Phi$-Sobolev inequalities \cite{CH1,CH2}.  

Furthermore, we consider the following operator-valued generalisation of matrix $\Phi$-entropy functionals:
\[
\bm{H}_\Phi(\bm{Z}) \triangleq \mathbb{E} \Phi(\bm{Z}) - \Phi(\mathbb{E}\bm{Z}).
\]
A  special case of this operator-valued $\Phi$-entropy functional is the operator-valued variance $\textbf{Var}(\bm{Z})$ defined in \cite{Tro15} and \cite{PMT14}, where $\Phi$ is the square function.
The equivalent conditions for the subadditivity under L\"owner partial ordering are  derived (Theorem~\ref{theo:operator}). In particular, we show that subadditivity of the operator-valued $\Phi$-entropies is equivalent to the convexity:
\[
\text{\textmd{Subadditvity} of } \bm{H}_{\Phi}(\bm{Z}) \; \Leftrightarrow\; \bm{H}_{\Phi}(\bm{Z}) \text{ is convex in } \bm{Z}.
\]
Our result directly yields the \emph{Operator Efron-Stein inequality}, which recovers the well-known Efron-Stein inequality \cite{ES81,Ste86} when random matrices reduce to real random variables.

\subsection{Our Results} \label{result}

\begin{table}[!h]
	\caption{Comparison between the equivalent characterisations of $\mathrm{\Phi}$-entropy functional class \textbf{(C1)} (Definition \ref{defn:C1}) and matrix $\mathrm{\Phi}$-entropy functional class \textbf{(C2)} (Definition \ref{defn:entropy})
	}\label{table:trace}
	\resizebox{\textwidth}{!}{
	\centering
		\begin{tabular}{c|l|l}
			\hline
			& \hspace*{\fill}\textbf{Classical $\mathrm{\Phi}$-Entropy Functional Class (C1)}\hspace*{\fill} & \hspace*{\fill}\textbf{Matrix $\mathrm{\Phi}$-Entropy Functional Class (C2)}\hspace*{\fill}\\
			\hline
			(a) & $\mathrm{\Phi}$ is affine or $\mathrm{\Phi}''>0$ and $1/\mathrm{\Phi}''$ is concave & $\mathrm{\Phi}$ is affine or $\mathsf{D}\mathrm{\Phi}'$ is invertible and $(\mathsf{D}\mathrm{\Phi}')^{-1}$ is concave\\
			(b) & convexity of $(u,v)\mapsto \mathrm{\Phi}(u+v) -\mathrm{\Phi}(u) - \mathrm{\Phi}'(u)v$ & convexity of $(\bm{u},\bm{v})\mapsto 
			\text{Tr}\left[\mathrm{\Phi}(\bm{u}+\bm{v}) -\mathrm{\Phi}(\bm{u}) - \mathsf{D}\mathrm{\Phi}[\bm{u}](\bm{v})\right]$ \\
			(c) & convexity of $(u,v)\mapsto (\mathrm{\Phi}'(u+v) - \mathrm{\Phi}'(u))v$ & convexity of $(\bm{u},\bm{v})\mapsto \text{Tr}\left[(\mathsf{D} \mathrm{\Phi}[\bm{u}+\bm{v}](\bm{v}) - \mathsf{D}\mathrm{\Phi}[\bm{u}](\bm{v})\right]$\\
			(d) & convexity of $(u,v)\mapsto \mathrm{\Phi}''(u) v^2$ & convexity of $(\bm{u},\bm{v})\mapsto \text{Tr}\left[\mathsf{D}^2\mathrm{\Phi}[\bm{u}](\bm{v},\bm{v})\right]$ \\
(e) & $\mathrm{\Phi}$ is affine or $\mathrm{\Phi}''>0$ and $\mathrm{\Phi}'''' \mathrm{\Phi}''\geq 2\mathrm{\Phi}'''^2$ & Equation \eqref{eq:(f)} \\
(f) & convexity of $(u,v)\mapsto t\mathrm{\Phi}(u) +(1-t)\mathrm{\Phi}(v)$ & convexity of $(\bm{u},\bm{v})\mapsto$ $\text{Tr} [t\mathrm{\Phi}(\bm{u}) +(1-t)\mathrm{\Phi}(\bm{v})$\\
&  $- \mathrm{\Phi}(t u+(1-t) v)$ for any $0\leq t\leq 1$ &  $- \mathrm{\Phi}(t\bm{u}+(1-t)\bm{v})]$ for any $0\leq t\leq 1$\\
(g) & $\mathbb{E}_1 H_\mathrm{\Phi}(Z|X_1) \geq H_\mathrm{\Phi}(\mathbb{E}_1 Z)$ & $\mathbb{E}_1 H_\mathrm{\Phi}(\bm{Z}|\bm{X}_1) \geq H_\mathrm{\Phi}(\mathbb{E}_1 \bm{Z})$\\
(h) & $H_\mathrm{\Phi}(Z)$ is a convex function of $Z$& $H_\mathrm{\Phi}(\bm{Z})$ is a convex function of $\bm{Z}$\\
(i) & $H_\mathrm{\Phi}(Z)=\sup_{T>0}\left\{\mathbb{E}\left[\Upsilon_1(T)\cdot Z+\Upsilon_2(T)\right] \right\}$ &  $H_\mathrm{\Phi}(\bm{Z})=\sup_{\bm{T}\succ\bm{0}}\left\{\tr\mathbb{E}\left[\bm{\Upsilon}_1(\bm{T})\cdot \bm{Z}+\bm{\Upsilon}_2(\bm{T})\right] \right\}$ \\
(j) & $H_\mathrm{\Phi}(Z)\leq \sum_{i=1}^n \mathbb{E} H_\mathrm{\Phi}^{(i)} (Z)$ & $H_\mathrm{\Phi}(\bm{Z})\leq \sum_{i=1}^n \mathbb{E} H_\mathrm{\Phi}^{(i)} (\bm{Z})$\\
\hline
\end{tabular}
	}
	\vspace*{-4pt}
\end{table}

\begin{table}[!h]
	\center
	\caption{Equivalent statements of the operator-valued $\mathrm{\Phi}$-entropy class \textbf{(C3)} (Definition \ref{defn:C3})
	}\label{table:operator}
	\resizebox{0.8\textwidth}{!}{
		\begin{tabular}{c|l}
			\hline
			& \hspace*{\fill}\textbf{Operator-Valued $\mathrm{\Phi}$-Entropy Class (C3)}\hspace*{\fill}\\
			\hline
			(a) & The second-order Fr\'echet derivative $\mathsf{D}^2\mathrm{\Phi}[\bm{u}](\bm{v},\bm{v})$ is jointly convex in $(\bm{u},\bm{v})$\\
			(b) & $\mathrm{\Phi}(\bm{u}+\bm{v}) -\mathrm{\Phi}(\bm{u}) - \mathsf{D}\mathrm{\Phi}[\bm{u}](\bm{v})$  is jointly convex in $(\bm{u},\bm{v})$\\
			(c) & $(\mathsf{D} \mathrm{\Phi}[\bm{u}+\bm{v}](\bm{v}) - \mathsf{D}\mathrm{\Phi}[\bm{u}](\bm{v})$  is jointly convex in $(\bm{u},\bm{v})$\\
			(d) & $\mathsf{D}^2\mathrm{\Phi}[\bm{u}](\bm{v},\bm{v})$  is jointly convex in $(\bm{u},\bm{v})$\\
			(e) & $t\mathrm{\Phi}(\bm{u}) +(1-t)\mathrm{\Phi}(\bm{v})- \mathrm{\Phi}(t \bm{u}+(1-t) \bm{v})$  is jointly convex in $(\bm{u},\bm{v})$
			for any $0\leq t\leq 1$\\
			(f) & $\mathbb{E}_1 \bm{H}_\mathrm{\Phi}(\bm{Z}|\bm{X}_1) \succeq \bm{H}_\mathrm{\Phi}(\mathbb{E}_1 \bm{Z})$\\
			(g) & $\bm{H}_\mathrm{\Phi}(\bm{Z})$ is a convex function of $\bm{Z}$\\
			(h) &  $\bm{H}_\mathrm{\Phi}(\bm{Z})=\sup_{\bm{T}\succ\bm{0}}\left\{\mathbb{E}\left[\bm{\Upsilon}_1(\bm{T})\cdot \bm{Z}+\bm{\Upsilon}_2(\bm{T})\right] \right\}$ \\
			(i) & $\bm{H}_\mathrm{\Phi}(\bm{Z})\preceq \sum_{i=1}^n \mathbb{E} \bm{H}_\mathrm{\Phi}^{(i)} (\bm{Z})$\\
			\hline
		\end{tabular}
	}
	\vspace*{-4pt}
\end{table}

We summarize our results here.
First, we derive equivalent characterisations for the matrix $\mathrm{\Phi}$-entropy functionals in Table~\ref{table:trace} (see Theorem \ref{theo:trace}). Notably, all known equivalent characterisations for the classical $\mathrm{\Phi}$-entropies can be generalised to their matrix correspondences.    
	We emphasise that additional characterisations of the $\Phi$-entropies prove to be  useful in many instances. The characterisations (b)-(d) in $\textbf{(C1)}$ are explored by Chafa\"i \cite{Cha06} to derive several entropic inequalities for M/M/$\infty$ queueing processes that are not diffusions.  With the characterisations (b)-(d), the difficulty of lacking the diffusion property can be circumvented and replaced by convexity. 
	Moreover, as shown in Corollary \ref{coro:unital}, item (f) in Table 2 can be used to demonstrate an interesting result in quantum information theory: the matrix $\Phi$-entropy functional of a quantum ensemble (i.e.~a set of quantum states with some prior distribution) is \emph{monotone} under any unital quantum channel. This property motivates us to study the dynamical evolution of a quantum ensemble and its mixing time, a fundamentally important problem in quantum computation (see our follow-up work \cite{CH2} for further details).

Second, we define and derive equivalent characterisations for operator-valued $\mathrm{\Phi}$-entropies in Table~\ref{table:operator} (see Theorem \ref{theo:operator}). Note that the only known statement in Table~\ref{table:trace} that is missing in Table~\ref{table:operator} is condition (e). In other words, we are not able to generalise (e) in Table~\ref{table:trace} to the non-commutative case.   
Finally, we employ the subadditivity of operator-valued $\Phi$-entropies to show the operator Efron-Stein inequality in Theorem \ref{theo:Efron}. 

\subsection{Prior Work} \label{prior}
For the history of the equivalent characterisations in the class $\textbf{(C1)}$, we refer to an excellent textbook \cite{BLM13} and the papers \cite{Cha03, Cha06}. 

The original definition of the matrix $\Phi$-entropy class; namely (a) in $\textbf{(C2)}$, is proposed by Chen and Tropp in 2014 \cite{CT14}. In the same paper, they also establish the subadditivity property (j) through (i) and (g): (a)$\Rightarrow$(i)$\Rightarrow$(g)$\Rightarrow$(j) in Table~\ref{table:trace}.  Shortly after,  the equivalent relation between (a) and the joint convexity of the matrix Br\'egman divergence (b) is proved in \cite{PV14}.  The equivalent relation between (a) and (d) is almost immediately implied by the result in \cite{HZ14} (see the detailed discussion in the proof of Theorem~\ref{theo:trace}).
The convexity of $H_{\Phi}(\bm{Z})$, (h), is noted in \cite{HZ14}. Here, we provide a transparent evidence---the joint convexity of (f).

We organise the paper in the following way. We collect necessary information of the Matrix Algebra in Section~\ref{Preliminaries}. The equivalent characterisations of matrix $\Phi$-entropy functionals are provided in Section \ref{sec_III}. We define the operator-valued $\Phi$-entropies and derive their equivalent statements in Section \ref{operator}.
Section \ref{Efron} shows an application of the subadditivity---the operator Efron-Stein inequality. The proofs of main results are collected in Sections \ref{proof:sub} and \ref{proof:sub_oper}, respectively.  Finally, we conclude the paper.

%
%

\section{Preliminaries} \label{Preliminaries}

We first introduce basic notation.

The set $\mathbb{M}^\textnormal{sa}$ refers to the subspace of self-adjoint operators on some separable Hilbert space. We denote by $\mathbb{M}^+$ (resp.~$\mathbb{M}^{++}$) the set  of positive semi-definite (resp.~positive-definite) operators in $\mathbb{M}^\textnormal{sa}$.
If the dimension $d$ of a Hilbert space  needs special attention, then we highlight it in subscripts, e.g.~$\mathbb{M}_d$ denotes the Banach space of $d\times d$ complex matrices. The trace function $\Tr:\mathbb{C}^{d\times d}\rightarrow \mathbb{C}$ is defined as the summation of eigenvalues. The normalised trace function $\tr$ for every $d\times d$ matrices $\bm{M}$ is denoted by
$\tr \left[\bm{M} \right] \triangleq  \frac1d \Tr\left[ \bm{M} \right]$.
For $p\in [1,\infty)$, the Schatten $p$-norm of an operator $\bm{M}$ is denoted as
$\|\bm{M}\| _{p} \triangleq  \left( \sum _{i } |\lambda_i(\bm{M})|^p \right)^{1/p}$,
where $\{\lambda_i(\bm{M})\}$ are the singular values of $\bm{M}$.
The Hilbert-Schmidt inner product is defined as $\langle \bm{A},\bm{B}\rangle \triangleq \Tr \bm{A}^\dagger \bm{B}$.
For $\bm{A},\bm{B}\in\mathbb{M}^\textnormal{sa}$, $\bm{A}\succeq \bm{B}$ means that $\bm{A}-\bm{B}$ is positive semi-definite. Similarly, $\bm{A} \succ \bm{B}$ means $\bm{A} - \bm{B}$ is positive-definite.
Throughout this paper, italic capital letters (e.g.~$\bm{X}$) are used to denote operators.


Denote a probability space $(\mathrm{\Omega},\mathrm{\Sigma},\mathbb{P})$. A random matrix $\bm{Z}$ defined on the probability space $(\mathrm{\Omega},\mathrm{\Sigma},\mathbb{P})$ means that it is a matrix-valued random variable defined on $\mathrm{\Omega}$. We denote the expectation of $\bm{Z}$ with respect to $\mathbb{P}$ by
\[
\mathbb{E}_{\mathbb{P}}[\bm{Z}] \triangleq \int_\Omega \bm{Z} \, \mathrm{d} \mu = \int_{x\in\Omega} \bm{Z}(x) \, \mathbb{P}(\mathrm{d}x),
\]
where the integral is the Bochner integral \cite{Die77, Mik78}.
We note that the results derived in this paper is universal for all probability spaces. Hence we will omit the subscript $\mathbb{P}$ of the expectation.
If we consider a sample space $\mathrm{\Omega}_1\times \mathrm{\Omega}_2$ with joint distribution $\mathbb{P}$. Then we denote the conditional expectation of $\bm{Z}$ with respect to the first space  $\mathrm{\Omega}_1$ by $\mathbb{E}_i[ \bm{Z} ] \triangleq \int_{x_1\in\Omega_1} \bm{Z}(x1, X_2 \cdot) \, \mathbb{P}_1(\mathrm{d}x_1)$, where $\mathbb{P}_1(x_1) = \int_{x_2\in\Omega_2} \mathbb{P}(x_1,\mathrm{d}x_2)$ is the marginal distribution on $\Omega_1$.

Let $\mathcal{U},\mathcal{W}$ be real Banach spaces.
The \emph{Fr\'{e}chet derivative} of a function $\mathcal{L}:\mathcal{U} \rightarrow \mathcal{W}$ at a point $\bm{X}\in\mathcal{U}$, if it exists\footnote{We assume the functions considered in the paper are Fr\'{e}chet differentiable. The readers can refer to, e.g.~\cite{Pel85,Bic12}, for conditions for when a function is Fr\'{e}chet differentiable. }, is a unique linear mapping $\mathsf{D}\mathcal{L}[\bm{X}]:\mathcal{U}\rightarrow\mathcal{W}$ such that
\[
\|\mathcal{L}(\bm{X}+\bm{E}) - \mathcal{L}(\bm{X}) - \mathsf{D}\mathcal{L}[\bm{X}](\bm{E})\|_{\mathcal{W}} = o(\|\bm{E}\|_\mathcal{U}),
\]
where $\|\cdot\|_{\mathcal{U}(\mathcal{W})}$ is a norm in $\mathcal{U}$ (resp.~$\mathcal{W}$). The notation $\mathsf{D}\mathcal{L}[\bm{X}](\bm{E})$ then is interpreted as ``the Fr\'{e}chet derivative of $\mathcal{L}$ at $\bm{X}$ in the direction $\bm{E}$''. 
The \emph{partial Fr\'{e}chet derivative} of multivariate functions can be defined as follows.
Let $\mathcal{U},\mathcal{V}$ and $\mathcal{W}$ be real Banach spaces, $\mathcal{L}:\mathcal{U}\times \mathcal{V} \rightarrow \mathcal{W}$. 
For a fixed $\bm{v}_0 \in\mathcal{V}$, $\mathcal{L}(\bm{u},\bm{v}_0)$ is a function of $\bm{u}$ whose derivative at $\bm{u}_0$, if it exists, is called the partial Fr\'{e}chet derivative of $\mathcal{L}$ with respect to $\bm{u}$, and is denoted by $\mathsf{D}_{\bm{u}}\mathcal{L} [\bm{u}_0, \bm{v}_0]$. The partial Fr\'{e}chet derivative $\mathsf{D}_{\bm{v}}\mathcal{L} [\bm{u}_0, \bm{v}_0]$ is defined similarly.
Similarly, the $m$-th Fr\'{e}chet derivative $\mathsf{D}^m \mathcal{L}[\bm{X}]$ is a unique multi-linear map from $\mathcal{U}^m\triangleq  \mathcal{U}\times\cdots\times\mathcal{U}$ ($m$ times) to $\mathcal{W}$ that satisfies
\begin{align*}
\|\mathsf{D}^{m-1}\mathcal{L}[\bm{X}+\bm{E}_m] (\bm{E}_1,\ldots,\bm{E}_{m-1}) &-
\mathsf{D}^{m-1}\mathcal{L}[\bm{X}] (\bm{E}_1,\ldots,\bm{E}_{m-1}) \\
&-  \mathsf{D}^m \mathcal{L} [\bm{X}](\bm{E}_1,\ldots,\bm{E}_{m}) \|_{\mathcal{W}}= o(\|\bm{E}_m \|_{\mathcal{U}})
\end{align*}
for each $\bm{E}_i \in \mathcal{U}, i=1,\ldots, m$.
The Fr\'{e}chet derivative enjoys several properties as in standard derivatives. We provide those facts in Appendix \ref{lemmas}.

A function $f:I\rightarrow \mathbb{R}$ is called \emph{operator convex} if for each $\bm{A},\bm{B}\in \mathbb{M}^\textnormal{sa}(I)$ and $0\leq t \leq 1$,
\[
f(t\bm{A}) + f((1-t)\bm{B}) \preceq f( t\bm{A} + (1-t)\bm{B}).
\]
Similarly, a function $f:I\rightarrow \mathbb{R}$ is called \emph{operator monotone} if for each $\bm{A}, \bm{B} \in\mathbb{M}^\textnormal{sa}(I)$,
\[
\bm{A}\preceq \bm{B} \; \Rightarrow \; f(\bm{A})\preceq f(\bm{B}).
\]

\subsection{Classical $\mathrm{\Phi}$-Entropy Functionals} \label{classical}

Let \textbf{(C1)} denote the class of functions $\mathrm{\Phi}:[0,\infty)\to \mathbb{R}$ that are continuous, convex on $[0,\infty)$, twice differentiable on $(0,\infty)$, and either $\mathrm{\Phi}$ is affine or $\mathrm{\Phi}''$ is strictly positive and $1/\mathrm{\Phi}''$ is concave.

\begin{defn}[Classical $\mathrm{\Phi}$-Entropies] \label{defn:C1} 
	Let $\Phi: [0,\infty)\to\mathbb{R}$ be a convex function.
	For every non-negative integrable random variable $Z$ so that $\mathbb{E}|Z|<\infty$ and $\mathbb{E}|\mathrm{\Phi}(Z)|<\infty$, the classical $\Phi$-entropy $H_{\Phi}(Z)$ is defined as
	\[
	H_\mathrm{\Phi}(Z) = \mathbb{E} \mathrm{\Phi}(Z) - \mathrm{\Phi}(\mathbb{E} Z). 
	\] 
	In particular, we are interested in $Z=f(X_1,\cdots,X_n)$, where $X_1,\cdots,X_n$ are independent random variables, and $f\geq 0$ is a measurable function.
\end{defn}

We say $H_\mathrm{\Phi}(Z)$ is \textit{subadditive} \cite{Led97} if 
\[
H_\mathrm{\Phi}(Z)\leq  \sum_{i=1}^n \mathbb{E} \left[H_\mathrm{\Phi}^{(i)}(Z)\right],
\]
where $H_\mathrm{\Phi}^{(i)}(Z) = \mathbb{E}_i \mathrm{\Phi}(Z) -\mathrm{\Phi}(\mathbb{E}_iZ)$ is the conditional $\Phi$-entropy, and $\mathbb{E}_i$ denotes conditional expectation conditioned on the $n-1$ random variables $X_{-i}\triangleq (X_1\cdots,X_{i-1},X_{i+1},\cdots,X_n)$. Sometimes we also denote $H_\mathrm{\Phi}^{(i)}(Z)$ by $H_\mathrm{\Phi}(Z|X_{-i})$.

It is a well-known result that, for any function $\mathrm{\Phi}\in \textbf{(C1)}$, $H_\mathrm{\Phi}(Z)$ is subadditive \cite[Corollary 3]{LO00} (see also \cite[Section 3]{BBM+05}). 

The following theorem establishes equivalent characterisations of classical $\Phi$-entropies. 
\begin{theo}[{\cite[Theorem 4.4]{Cha06}}] \label{thm_classical}
	The following statements are equivalent.
	\begin{enumerate}
		\item[(a)] $\mathrm{\Phi}\in\textnormal{\textbf{(C1)}}$: $\mathrm{\Phi}$ is affine or $\mathrm{\Phi}''>0$ and $1/\mathrm{\Phi}''$ is concave;
		\item[(b)] Br\`egman divergence $(u,v)\mapsto \mathrm{\Phi}(u+v) -\mathrm{\Phi}(u) - \mathrm{\Phi}'(u)v$ is convex;
		\item[(c)] $(u,v)\mapsto (\mathrm{\Phi}'(u+v) - \mathrm{\Phi}'(u))v$ is convex;
		\item[(d)] $(u,v)\mapsto \mathrm{\Phi}''(u) v^2$ is convex;
		\item[(e)] $\mathrm{\Phi}$ is affine or $\mathrm{\Phi}''>0$ and $\mathrm{\Phi}'''' \mathrm{\Phi}''\geq 2\mathrm{\Phi}'''^2$;
		\item[(f)] $(u,v)\mapsto t\mathrm{\Phi}(u) +(1-t)\mathrm{\Phi}(v) - \mathrm{\Phi}(t u+(1-t) v)$ is convex for any $0\leq t\leq1$;
		\item[(g)] $\mathbb{E}_1 H_\mathrm{\Phi}({Z}|{X}_1) \geq H_\mathrm{\Phi}(\mathbb{E}_1 {Z})$;
		\item[(h)] $\{H_\mathrm{\Phi}(Z)\}_{\mathrm{\Phi}\in\textnormal{\textbf{(C1)}}}$ forms a convex set;
		\item[(i)] $H_\mathrm{\Phi}(Z)=\sup_{T>0}\left(\mathbb{E}\left[(\mathrm{\Phi}'(T)-\mathrm{\Phi}'(\mathbb{E}\,T))(Z-T)\right]+H_\mathrm{\Phi}(T) \right)$;
		\item[(j)] $H_\mathrm{\Phi}(Z)\leq \sum_{i=1}^n \mathbb{E} H_\mathrm{\Phi}^{(i)} (Z)$.
	\end{enumerate}
\end{theo}

\section{Equivalent Characterizations of Matrix $\mathrm{\Phi}$-Entropy Functionals}\label{sec_III}

In this section, we first introduce matrix $\Phi$-entropy functionals, and present the main result (Theorem~\ref{theo:trace}) of this section; namely, new characterisations of the matrix $\Phi$-entropy functionals.

Chen and Tropp introduce the class of matrix $\mathrm{\Phi}$-entropies, and prove its subadditivity in 2014 \cite{CT14}. In this section, we will show that all equivalent characterisations of classical $\Phi$-entropies in Theorem~\ref{thm_classical} have a one-to-one correspondence for the class of matrix $\mathrm{\Phi}$-entropies.

Let $d$ be a natural number. The class $\mathrm{\Phi}_d$ contains each function $\mathrm{\Phi}:(0,\infty)\rightarrow \mathbb{R}$ that is either affine or satisfies the following three conditions:
\begin{enumerate}
	\item $\mathrm{\Phi}$ is convex and continuous at zero.
	\item $\mathrm{\Phi}$ is twice continuously differentiable.
	\item Define $\mathrm{\Psi}(t)=\mathrm{\Phi}'(t)$ for $t>0$. The Fr\'{e}chet derivative $\mathsf{D}\mathrm{\Psi}$ of the standard matrix  function ${\mathrm{\Psi}}:\mathbb{M}_d^{++}\rightarrow \mathbb{M}_d^{\textnormal{sa}}$ is an invertible linear map on $\mathbb{M}_d^{++}$, and the map $\bm{A}\mapsto (\mathsf{D}\mathrm{\Psi}[\bm{A}])^{-1}$ is concave with respect to the L\"{o}wner partial ordering on positive definite matrices.
\end{enumerate}
Define $\textbf{(C2)} \triangleq  \mathrm{\Phi}_\infty \equiv \bigcap_{d=1}^\infty \mathrm{\Phi}_d$.

\begin{defn}[Matrix $\mathrm{\Phi}$-Entropy Functional \cite{CT14}] \label{defn:entropy}
	Let $\Phi:[0,\infty)\to\mathbb{R} $ be a convex function.
	Consider a random matrix $\bm{Z}\in\mathbb{M}_d^{+}$ with $\mathbb{E}\|\bm{Z}\|_\infty<\infty$ and $\mathbb{E}\|\mathrm{\Phi}(\bm{Z})\|_\infty<\infty$.
	The matrix $\mathrm{\Phi}$-entropy $H_\mathrm{\Phi}(\bm{Z})$ is defined as
	\[ 
	H_\mathrm{\Phi}(\bm{Z})\triangleq \tr\left[\mathbb{E}\mathrm{\Phi}(\bm{Z})-\mathrm{\Phi}(\mathbb{E}\,\bm{Z})\right].
	\]
	The corresponding conditional matrix $\Phi$-entropy can be defined under the $\sigma$-algebra.
\end{defn}

\begin{theo}[Subadditivity of Matrix $\mathrm{\Phi}$-Entropy Functional {\cite[Theorem 2.5]{CT14}}] \label{theo:sub} Let $\Phi\in$ \textnormal{\textbf{(C2)}}, and assume $\bm{Z}$ is a measurable function of $(\bm{X}_1,\ldots,\bm{X}_n)$.
	\begin{eqnarray}\label{eq:entropy}
	H_\mathrm{\Phi}(\bm{Z})\leq \sum_{i=1}^n \mathbb{E} \Big[H^{(i)}_\mathrm{\Phi}(\bm{Z}) \Big],
	\end{eqnarray}
	where $H_\mathrm{\Phi}^{(i)}(\bm{Z}) = \mathbb{E}_i \mathrm{\Phi}(\bm{Z}) -\mathrm{\Phi}(\mathbb{E}_i\bm{Z})$ is the conditional entropy, and $\mathbb{E}_i$ denotes conditional expectation conditioned on the $n-1$ random matrices $\bm{X}_{-i}\triangleq (\bm{X}_1,\ldots,\bm{X}_{i-1},\bm{X}_{i+1},\ldots,\bm{X}_n)$.
\end{theo}

The following theorem is the main result of this section. We show that all the equivalent conditions in Theorem~\ref{thm_classical} also hold for the class of matrix $\mathrm{\Phi}$-entropy functionals. 
Hence, we have a much comprehensive understanding on the class of matrix $\mathrm{\Phi}$-entropy functionals 

\begin{theo}\label{theo:trace}
	The following statements are equivalent. 
	\begin{enumerate}
		\item[(a)] $\Phi\in (\textnormal{\textbf{C2}})$: $\mathrm{\Phi}$ is affine or $\mathsf{D}\mathrm{\Psi}$ is invertible and $\bm{A}\mapsto (\mathsf{D}\mathrm{\Psi}[\bm{A}])^{-1}$ is operator concave;
		\item[(b)] Matrix Br\`egman divergence: $(\bm{A},\bm{B})\mapsto \Tr[\mathrm{\Phi}(\bm{A}+\bm{B}) -\mathrm{\Phi}(\bm{A}) - \mathsf{D}\mathrm{\Phi}[\bm{A}](\bm{B})]$ is convex;
		\item[(c)] $(\bm{A},\bm{B})\mapsto \Tr[\mathsf{D}\mathrm{\Phi}[\bm{A}+\bm{B}](\bm{B}) - \mathsf{D}\mathrm{\Phi}[\bm{A}](\bm{B})]$ is convex;
		\item[(d)] $(\bm{A},\bm{B})\mapsto \Tr[\mathsf{D}^2\mathrm{\Phi}[\bm{A}] (\bm{B},\bm{B})]$ is convex;
		\item[(e)] $\mathrm{\Phi}$ is affine or $\mathrm{\Phi}''>0$ and 
		\begin{align} \label{eq:(f)}
		&\Tr\left[ \bm{h} \cdot \left(\mathsf{D}\mathrm{\Psi}[\bm{A}]\right)^{-1} \circ \mathsf{D}^3\mathrm{\Psi}[\bm{A}]\left(\bm{k},\bm{k}, \left( \mathsf{D}\mathrm{\Psi}[\bm{A}]\right)^{-1}(\bm{h}) \right) \right]\\
		&\geq 2\Tr \left[ \bm{h} \cdot \left(\mathsf{D}\mathrm{\Psi}[\bm{A}]\right)^{-1} \circ \mathsf{D}^2\mathrm{\Psi}[\bm{A}]\left(\bm{k}, 
		\left(\mathsf{D}\mathrm{\Psi}[\bm{A}]\right)^{-1}
		\left(  \mathsf{D}^2\mathrm{\Psi}[\bm{A}]\left(\bm{k}, \left(\mathsf{D}\mathrm{\Psi}[\bm{A}]\right)^{-1}(\bm{h}) \right)    \right)
		\right)
		\right],\notag
		\end{align}
		for each $\bm{A}\succeq \bm{0}$ and $\bm{h},\bm{k}\in\mathbb{M}_d^\textnormal{sa}$;
		\item[(f)] $(\bm{A},\bm{B})\mapsto$ $\Tr [t\mathrm{\Phi}(\bm{A}) +(1-t)\mathrm{\Phi}(\bm{B})- \mathrm{\Phi}(t \bm{A}+(1-t) \bm{B})]$ is convex for any $0\leq t\leq 1$;
		\item[(g)] $\mathbb{E}_1 H_\mathrm{\Phi}(\bm{Z}|\bm{X}_1) \geq H_\mathrm{\Phi}(\mathbb{E}_1 \bm{Z})$;
		\item[(h)] $\{H_\mathrm{\Phi}(\bm{Z})\}_{\mathrm{\Phi}\in\textnormal{\textbf{(C2)}}}$ forms a convex set of convex functions;
		\item[(i)] $H_\mathrm{\Phi}(\bm{Z})=\sup_{\bm{T}\succ\bm{0}}\left\{\tr\mathbb{E}\left[(\mathrm{\Phi}'(\bm{T})-\mathrm{\Phi}'(\mathbb{E}\bm{T}))(\bm{Z}-\bm{T})\right]+H_\mathrm{\Phi}(\bm{T}) \right\}$;
		\item[(j)] $H_\mathrm{\Phi}(\bm{Z})\leq \sum_{i=1}^n \mathbb{E} H_\mathrm{\Phi}^{(i)} (Z)$.
	\end{enumerate}
\end{theo}
We note that the statements $(a) \Rightarrow (i) \Rightarrow (g) \Rightarrow (j)$ was proved by Chen and Tropp in \cite{CT14}.
The equivalence of $(a) \Leftrightarrow (b)$ was shown in \cite[Theorem 2]{PV14}. Hansen and Chang established an equivalence of item $(a)$ and the convexity of the following map:
\begin{align} \label{eq:Hanzen}
(\bm{A},\bm{X}) \mapsto \langle \bm{X}, \mathsf{D}\Phi'[\bm{A}](\bm{X}).
\end{align}
From Lemma \ref{lemm:trace_Petz}, it is not hard to observe that Eq.~\eqref{eq:Hanzen} is equivalent to item (d), i.e.~
\begin{align*} 
\Tr ( \mathsf{D}^2 \mathrm{\Phi} [ \bm{A} ] ( \bm{X , X} ) ) 
&= \left\langle \bm{X}, \mathsf{D}\mathrm{\Phi}' [ \bm{A} ] ( \bm{X} ) \right\rangle.
\end{align*}  
We provide the detailed proof of the remaining equivalence statements in Section \ref{proof:sub}. 


\section{Operator-Valued $\Phi$-Entropies} \label{operator}
In this section, we extend the notion of matrix $\Phi$-entropy functionals (i.e.~real-valued) to operator-valued $\Phi$-entropies.


\begin{defn}[Operator-Valued Entropy Class] \label{defn:C3}
	Let $d$ be a natural number. The class $\mathrm{\Phi}_d$ contains each function $\mathrm{\Phi}:[0,\infty)\rightarrow \mathbb{R}$ such that its second-order Fr\'{e}chet derivative exists and the following map satisfies the joint convexity (under the L\"{o}wner partial ordering) condition:
	\begin{eqnarray} \label{eq:entropy_class}
	(\bm{A},\bm{X})\mapsto\mathsf{D}^2\mathrm{\Phi}[\bm{A}](\bm{X},\bm{X}), \quad \forall \bm{A}\in\mathbb{M}_d^\textnormal{sa} \quad \bm{X}\in\mathbb{M}_d.
	\end{eqnarray}
	We denote the class of operator-valued $\mathrm{\Phi}$-entropies $\mathrm{\Phi}_\infty \triangleq \bigcap_{d=1}^\infty \mathrm{\Phi}_d$ by \textbf{(C3)}.
\end{defn}

\begin{defn}[Operator-Valued $\Phi$-Entropies] \label{defn:operator_entropy}
	Let $\Phi:[0,\infty) \to \mathbb{R}$ be a convex function.
	Consider a random matrix $\bm{Z}$ taking values in $\mathbb{M}^{+}$, with $\mathbb{E}\|\bm{Z}\|_\infty<\infty$ and $\mathbb{E}\|\mathrm{\Phi}(\bm{Z})\|_\infty<\infty$.
	That is, the random matrix $\bm{Z}$ and $\Phi(\bm{Z})$ are Bochner integrable \cite{Die77,Mik78} (Hence $\mathbb{E}\bm{Z}$ and $\mathbb{E}\Phi(\bm{Z})$ exist and are well-defined).
	The operator-valued $\mathrm{\Phi}$-entropy $\bm{H}_\mathrm{\Phi}$ is defined as
	\begin{eqnarray} \label{eq:operator_entropy}
	\bm{H}_\mathrm{\Phi}(\bm{Z})\triangleq\mathbb{E}\mathrm{\Phi}(\bm{Z})-\mathrm{\Phi}(\mathbb{E}\bm{Z}).
	\end{eqnarray}
 	The corresponding conditional terms can be defined under the $\sigma$ algebra.
\end{defn}

It is worth mentioning that the matrix $\mathrm{\Phi}$-entropy functional \cite{CT14} in Section \ref{sec_III} is non-negative for every convex function $\Phi$ due to the fact that the trace function $\tr\mathrm{\Phi}$ is also convex \cite{von55} (or see e.g.~\cite[Sec.~2.2]{Car09}).
However, according to the operator Jensen's inequality \cite[Theorem 3.2]{FZ07}, only the operator convex function $\Phi$ ensures the operator-valued $\mathrm{\Phi}$-entropy non-negative. 

In the following, we show that the the entropy class \textbf{(C3)} is not an empty set. 
\begin{prop} \label{prop:square}
	The square function $\Phi(u)=u^2$ belongs to \textnormal{\textbf{(C3)}}.
\end{prop}
\begin{proof}
	It suffices to verify the joint convexity of the map:
	\begin{align*}
	\mathsf{D}^2 \mathrm{\Phi}[\bm{A}]\left( \bm{X}, \bm{X} \right) = 2\bm{X}^2 \quad \text{for all } \bm{A},\bm{X}\in\mathbb{M}^+,
	\end{align*}
	where we use the identity of second-order Fr\'{e}chet derivative (see e.g.~\cite[Example X.4.6]{Bha97}) of the square function.
	Since the square function is operator convex, hence
	$\mathrm{\Phi}(u)=u^2$ belongs to the operator-valued $\mathrm{\Phi}$-entropy class \textbf{(C3)}. 
\end{proof}

\subsection{Subaddtivity of operator-valued $\mathrm{\Phi}$-entropies} \label{sub}

Denote by $\bm{X}\triangleq(\bm{X}_1,\ldots,\bm{X}_n)$ a series of independent random variables taking values in a Polish space, and let $\bm{X}_{-1}$
\[
\bm{X}_{-i}\triangleq(\bm{X}_1,\ldots,\bm{X}_{i-1},\bm{X}_{i+1},\ldots,\bm{X}_n).
\]
Let a positive semi-definite matrix $\bm{Z}$ that depends on the series of random variables $\bm{X}$:
\[
\bm{Z}\triangleq\bm{Z}(\bm{X}_1,\ldots,\bm{X}_n)\in\mathbb{M}^+.
\]
Throughout this paper, we assume the random matrix $\bm{Z}$ satisfies the integrability conditions: $|\bm{Z}|\triangleq \sqrt{\bm{Z}^2}$ and $|\Phi(\bm{Z})|$ is Bochner integrable for $\Phi\in\textbf{(C3)}$.

\begin{theo}
	[Subadditivity of Operator-Valued $\mathrm{\Phi}$-Entropy] \label{theo:sub_oper}
	Fix a function $\mathrm{\Phi}\in\textnormal{\textbf{(C3)}}$. Under the prevailing assumptions,
	\begin{eqnarray} \label{eq:sub}
	\bm{H}_\mathrm{\Phi}(\bm{Z}) \preceq \sum_{i=1}^n \mathbb{E}\left[\bm{H}_\mathrm{\Phi}^{(i)}(\bm{Z})\right],
	\end{eqnarray}
	where $\bm{H}_\Phi^{(i)} (\bm{Z}) = \bm{H}_\Phi(\bm{Z}|\bm{X}_{-i})\triangleq \mathbb{E}_i \Phi(\bm{Z}) - \Phi(\mathbb{E}_i\bm{Z})$.
\end{theo}
The proof is given in Section \ref{proof:sub_oper}.

\subsection{Equivalent characterisations of operator-valued $\mathrm{\Phi}$-entropies} \label{equivalent}

In this section, we derive alternative characterisations of the class \textbf{(C3)} in Theorem \ref{theo:operator}. As an application of the entropy class, we show that if the function $\Phi$ belongs to \textbf{(C3)}, then the operator-valued $\Phi$-entropy is monotone under any {unital completely positive map}. 

\begin{theo}\label{theo:operator}
	The following statements are equivalent. 
	\begin{enumerate}
		\item[(a)] $\Phi\in \textnormal{\textbf{(C3)}}$: convexity of $(\bm{A},\bm{B})\mapsto 
		\mathsf{D}^2\mathrm{\Phi}[\bm{A}](\bm{B},\bm{B})$;
		\item[(b)] Operator-valued Br\`egman divergence: $(\bm{A},\bm{B})\mapsto \mathrm{\Phi}(\bm{A}+\bm{B}) -\mathrm{\Phi}(\bm{A}) - \mathsf{D}\mathrm{\Phi}[\bm{A}](\bm{B})$ is convex;
		\item[(c)] $(\bm{A},\bm{B})\mapsto \mathsf{D}\mathrm{\Phi}[\bm{A}+\bm{B}](\bm{B}) - \mathsf{D}\mathrm{\Phi}[\bm{A}](\bm{B})$ is convex;
		\item[(d)] $(\bm{A},\bm{B})\mapsto \mathsf{D}^2\mathrm{\Phi}[\bm{A}] (\bm{B},\bm{B})$ is convex;
		\item[(e)] Convexity of $(\bm{A},\bm{B})\mapsto$ $t\mathrm{\Phi}(\bm{A}) +(1-t)\mathrm{\Phi}(\bm{B})- \mathrm{\Phi}(t \bm{A}+(1-t) \bm{B})$ for any $0\leq t\leq 1$;
		\item[(f)] $\mathbb{E}_1 \bm{H}_\mathrm{\Phi}(\bm{Z}|\bm{X}_1) \succeq \bm{H}_\mathrm{\Phi}(\mathbb{E}_1 \bm{Z})$;
		\item[(g)] $\{\bm{H}_\mathrm{\Phi}(\bm{Z})\}_{\mathrm{\Phi}\in\textnormal{\textbf{(C3)}}}$ forms a convex set of convex functions;
		\item[(h)] $\bm{H}_\mathrm{\Phi}(\bm{Z})=\sup_{\bm{T}\succ\bm{0}}\left\{\mathbb{E}\left[\mathsf{D}\Phi[\bm{T}](\bm{Z}-\bm{T})-\mathsf{D}\mathrm{\Phi}[\mathbb{E}\bm{T}](\bm{Z}-\bm{T})\right]+\bm{H}_\mathrm{\Phi}(\bm{T}) \right\}$;
		\item[(i)] $\bm{H}_\mathrm{\Phi}(\bm{Z})\preceq \sum_{i=1}^n \mathbb{E} \bm{H}_\mathrm{\Phi}^{(i)} (\bm{Z})$.
	\end{enumerate}
\end{theo}
%
The proof is omitted since it directly follows from that of Theorem \ref{theo:trace} without taking traces. 

\begin{remark}
	In item (g) of Theorem \ref{theo:operator}, we introduce a supremum representation for the operator-valued $\Phi$-entropies. The supremum is defined as the least upper bound (under L\"owner partial ordering) among the set of operators. In general, the supremum might not exist due to matrix partial ordering; however, the supremum in (g) exists and is attained when $\bm{T} \equiv \bm{Z}$. 
\Endremark
\end{remark}

In the following, we demonstrate a monotone property of operator-valued $\Phi$-entropies when $\Phi\in\textbf{(C3)}$.

\begin{prop} \label{prop:unital}
	[Monotonicity of Operator-Valued {$\Phi$}-Entropies]
	Fix a convex function $\Phi \in \textbf{(C3)}$, then the operator-valued $\Phi$-entropy $\bm{H}_\Phi(\bm{Z})$ is monotone under any unital completely positive map $\mathsf{N}$, i.e.~
	\[
	\bm{H}_\Phi( \mathsf{N}(\bm{Z})  ) \preceq \bm{H}_\Phi( \bm{Z}  )
	\]
	for any random matrix $\bm{Z}$ taking values in $\mathbb{M}^+$.
\end{prop}
\begin{proof}
	If $\Phi\in\textnormal{\textbf{(C3)}}$, by item (e) in Theorem \ref{theo:operator}, we have the joint convexity of the map:
	\[
	\bm{F}_t(\bm{A},\bm{B}) \triangleq t\mathrm{\Phi}(\bm{A}) +(1-t)\mathrm{\Phi}(\bm{B})- \mathrm{\Phi}(t \bm{A}+(1-t) \bm{B})
	\]
	for any $0\leq t\leq 1$.
	Let $\bm{X} = (\bm{A},\bm{B})$ denote the pair of matrices.

For any completely positive unital  map $\mathsf{N}$, it can be expressed in the following form (see e.g.~\cite{MW09}):
\[
\mathsf{N}(\bm{A}) = \sum_{i} \bm{K}_i \bm{A} \bm{K}_i^\dagger,
\]
where $\sum_i \bm{K}_i \bm{K}_i^\dagger = \bm{I}$ (the identity matrix in $\mathbb{M}^\text{sa}$), and $^\dagger$ denotes the complex conjugate.
Hence, by Jensen's operator inequality, Proposition \ref{prop:Jensen}, yields
	\[
	\bm{F}_t( \mathsf{N}(\bm{X})  ) \preceq \bm{F}_t( \bm{X}  ),\quad \forall 0\leq t \leq 1
	\]
	for any completely positive unital  map $\mathsf{N}$, which implies the monotonicity of $\bm{H}_\Phi(\bm{Z})$.
\end{proof}
Following the same argument, the matrix $\Phi$-entropy functional satisfies the monotone property if $\Phi\in\textbf{(C2)}$.

\begin{coro} \label{coro:unital}
	[Monotonicity of Matrix {$\Phi$}-Entropy Functionals]
	Fix a convex function $\Phi \in \textbf{(C2)}$, then the matrix $\Phi$-entropy functional ${H}_\Phi(\bm{Z})$ is monotone under any unital completely positive map $\mathsf{N}$:
	$
	{H}_\Phi( \mathsf{N}(\bm{Z})  ) \leq {H}_\Phi( \bm{Z}  )
	$
	for any random matrix $\bm{Z}$ taking values in $\mathbb{M}^+$.
\end{coro}

We remark that the monotonicity of a quantum ensemble is only known for when $\Phi(x) = x \log x$. This is the famous result in quantum information theory, namely, the monotone property of the Holevo quantity \cite{Pet03}. Our Corollary~\ref{coro:unital} extends the monotonicity of a quantum ensemble to any function $\Phi \in \textbf{(C2)}$.

\section{Applications: Operator Efron-Stein Inequality} \label{Efron}

In this section, we employ the operator subadditivity of $\bm{H}_{u\to u^2} (\bm{Z})$ to prove the operator Efron-Stein inequality.
For $1\leq i\leq n$, let $\bm{X}_1',\ldots,\bm{X}_n'$ be independent copies of $\bm{X}_1,\ldots,\bm{X}_n$, and denote $\widetilde{\bm{X}}^{(i)}\triangleq (\bm{X}_1,\ldots,\bm{X}_{i-1},\bm{X}_i'
,\bm{X}_{i+1},\ldots,\bm{X}_n)$, i.e.~replacing the $i$-th component of $\bm{X}$ by the independent copy $\bm{X}_i'$.

Define the quantity\footnote{Note that we will use notation $\bm{\mathcal{E}}(\mathcal{L})$ and $\bm{\mathcal{E}}(\bm{Z})$ interchangeably.}
\[
\bm{\mathcal{E}}(\mathcal{L}) \triangleq  \frac12 \, \mathbb{E}\left[\sum_{i=1}^n \left( \mathcal{L}(\bm{X})-\mathcal{L}\left(\widetilde{\bm{X}}^{(i)}\right)\right)^2\right],
\]
and denote the operator-valued variance  of a random matrix $\bm{A}$ (taking values in $\mathbb{M}^{sa}$) by
\[
\textbf{Var}(\bm{A}) \triangleq  \mathbb{E}\left( \bm{A} - \mathbb{E}\bm{A} \right)^2 = \mathbb{E}\bm{A}^2 - \left(\mathbb{E}\bm{A}\right)^2 .
\]

\begin{theo}
[Operator Efron-Stein Inequality] \label{theo:Efron}
	With the prevailing assumptions, we have
\begin{align*}
\textnormal{\textbf{Var}}(\bm{Z}) &\preceq 
 \bm{\mathcal{E}}(\bm{Z}).
	\end{align*}
\end{theo}

\begin{proof}
This theorem is a direct consequence of the subadditivity of operator-valued $\mathrm{\Phi}$-entropies; namely, Theorem \ref{theo:sub_oper} with $\mathrm{\Phi}(u)=u^2$.

For two independent and identical random matrices $\bm{A}$, $\bm{B}$, direct calculation yields:
	\begin{align*}
	\frac12 \mathbb{E} \left[ \left( \bm{A}-\bm{B} \right)^2\right]
	&= \frac12 \mathbb{E} \left[ \bm{A}^2-\bm{A}\bm{B}-\bm{B}\bm{A} + \bm{B}^2 \right]\\
	&= 		\textnormal{\textbf{Var}}\left(\bm{A}\right).
	\end{align*}
Observe that  $\bm{Z}_i'$ is an independent copy of $\bm{Z}$ conditioned on $\bm{X}_{-i}$. Denote $\textnormal{\textbf{Var}}^{(i)}\left(\bm{Z}\right)\triangleq  \mathbb{E}_i \left( \bm{Z}-\mathbb{E}_i \bm{Z} \right)^2$ for all $i=1,\ldots,n$. Then

	\begin{align*}
	\textnormal{\textbf{Var}}^{(i)}\left(\bm{Z}\right) 
	&= \frac12 \, \mathbb{E}_i\left[\left(\bm{Z}-\bm{Z}_i'\right)^2\right].  
	\end{align*}
Finally, Theorem \ref{theo:sub_oper} and Proposition \ref{prop:square} lead to
	\begin{align*}
	\textnormal{\textbf{Var}}(\bm{Z}) &
	=\bm{H}_{u\mapsto u^2}(\bm{Z}) \\
	&\preceq \sum_{i=1}^n \mathbb{E} \, \bm{H}_{u\mapsto u^2}^{(i)} \left(\bm{Z} \right) \\
	& =  \sum_{i=1}^n \textnormal{\textbf{Var}}^{(i)}\left(\bm{Z}\right) \\
	&=\bm{\mathcal{E}(\bm{Z})}.
	\end{align*}
\end{proof}

Note that the established operator Efron-Stein inequality directly leads to a \emph{matrix polynomial Efron-Stein inequality}.

\begin{coro}[Matrix Polynomial Efron-Stein] \label{theo:poly_Efron}
	With the prevailing assumptions, for each natural number $p\geq 1$, we have
	\begin{align*}
	\left\| \mathbb{E} \left( \bm{Z} - \mathbb{E} \bm{Z} \right)^2 \right\|_{p}^{p} 
	&\leq 
	\left\|  \frac12 \sum_{i=1}^n \mathbb{E}\left[ \left( \bm{Z} -\bm{Z}_i' \right) ^2 \right] \right\|_{p}^{p}.
	\end{align*}
\end{coro}

Corollary~\ref{theo:poly_Efron} is a variant of the matrix polynomial Efron-Stein inequality derived in \cite[Theorem 4.2]{PMT14}. 


\section{Proof of Theorem {\ref{theo:trace}}} \label{proof:sub}
\begin{proof} 
	
	\begin{description}
		
		\item[$(a) \Rightarrow (i) \Rightarrow (g) \Rightarrow (j)$] This statement is proved by Chen and Tropp in \cite{CT14}.

		\item[$(a) \Leftrightarrow (b)$] This equivalent statement is proved in \cite[Theorem 2]{PV14}.
		
		\item[$(a) \Leftrightarrow (d)$] 
		Theorem 2.1 in \cite{HZ14} proved the equivalence of $(a)$ and the following convexity lemma. 
		\begin{lemm}[Convexity Lemma {\cite[Lemma 4.2]{CT14}}] \label{lemm:convex}
			Fix a function $\mathrm{\Phi}\in \textnormal{\textbf{(C2)}}$, and let $\mathrm{\Psi}=\mathrm{\Phi}'$. Suppose that $\bm{A}$ is a random matrix taking values in $\mathbb{M}_d^{++}$, and let $\bm{X}$ be a random matrix taking values in $\mathbb{M}_d^{sa}$. Assume that $\|\bm{A}\|,\,\|\bm{X}\|$ are integrable. Then
			\[
			\mathbb{E}\left\langle \bm{X}, \mathsf{D}\mathrm{\Psi}[\bm{A}](\bm{X})\right\rangle \geq \left\langle \mathbb{E}[\bm{X}],\mathsf{D}\mathrm{\Psi}[\mathbb{E}\bm{A}](\mathbb{E}\,\bm{X})\right\rangle.
			\]
		\end{lemm}
		What remains is to establish equivalence between the convexity lemma and condition $(d)$. This follows easily from Lemma \ref{lemm:trace_Petz}:
		\begin{align*} 
		\Tr ( \mathsf{D}^2 \mathrm{\Phi} [ \bm{A} ] ( \bm{X , X} ) ) 
		&= \left\langle \bm{X}, \mathsf{D}\mathrm{\Phi}' [ \bm{A} ] ( \bm{X} ) \right\rangle,
		\end{align*}

		\begin{remark}
			In Ref.~\cite[Lemma 4.2]{CT14}, it is shown that the concavity of the map: 
			\[
			\bm{A} \mapsto \left\langle \bm{X} \left( \mathsf{D}\mathrm{\Psi} [\bm{A}] \right)^{-1}(\bm{X}) \right\rangle, \quad \forall \bm{X}\in\mathbb{M}_d^{sa}
			\]
			implies the joint convexity of the map (i.e.~Lemma \ref{lemm:convex}).
			\begin{align} \label{eq:convex}
			(\bm{X},\bm{A})\mapsto \left\langle \bm{X} \left( \mathsf{D}\mathrm{\Psi} [\bm{A}] \right)(\bm{X}) \right\rangle.
			\end{align}
		\end{remark}


		\item[$(b) \Leftrightarrow (c) \Leftrightarrow (d)$] 
		
		Define ${A}_\mathrm{\Phi}, {B}_\mathrm{\Phi}, {C}_\mathrm{\Phi}:\mathbb{M}_d^+ \times \mathbb{M}_d^+\rightarrow \mathbb{R}$ as
		\begin{align*}
		{A}_\mathrm{\Phi} (\bm{u},\bm{v}) &\triangleq \Tr[ \mathrm{\Phi}(\bm{u+v}) - \mathrm{\Phi}(\bm{u}) - \mathsf{D}\mathrm{\Phi}[\bm{u}](\bm{v})]\\
		{B}_\mathrm{\Phi} (\bm{u},\bm{v}) &\triangleq \Tr[ \mathsf{D}\mathrm{\Phi}[\bm{u+v}](\bm{v}) - \mathsf{D}\mathrm{\Phi}[\bm{u}](\bm{v})]\\
		{C}_\mathrm{\Phi} (\bm{u},\bm{v}) &\triangleq \Tr[ \mathsf{D}^2\mathrm{\Phi}[\bm{u}](\bm{v},\bm{v})].
		\end{align*}
		Following from \cite{Cha06}, we can establish the following relations: for any $(\bm{u},\bm{v})\in\mathbb{M}_d^+ \times \mathbb{M}_d^+$,
		\begin{eqnarray}
		{A}_\mathrm{\Phi}(\bm{u},\bm{v})&=&\int_0^1 (1-s){C}_\mathrm{\Phi}(\bm{u}+s\bm{v},\bm{v}) \mathrm{d} s \label{eq:A_int}\\
		{B}_\mathrm{\Phi}(\bm{u},\bm{v})&=&\int_0^1{C}_\mathrm{\Phi}(\bm{u}+s\bm{v},\bm{v}) \mathrm{d} s,\label{eq:B_int}
		\end{eqnarray}
		and for small enough $\epsilon>0$,
		\begin{eqnarray}
		{A}_\mathrm{\Phi}(\bm{u},\epsilon\bm{v})&=&\frac12 {C}_\mathrm{\Phi}(\bm{u},\bm{v})\epsilon^2 + o(\epsilon^2); \label{eq:A_diff}\\
		{B}_\mathrm{\Phi}(\bm{u},\epsilon\bm{v})&=& {C}_\mathrm{\Phi}(\bm{u},\bm{v})\epsilon^2 + o(\epsilon^2). \label{eq:B_diff}
		\end{eqnarray}
		Eq.~\eqref{eq:A_int} is exactly the integral representation for the matrix Br\'egman divergence proved in  \cite{PV14}.
		Similarly, Eq.~\eqref{eq:B_int} follows from
		\begin{align*}
		{B}_\mathrm{\Phi} (\bm{u},\bm{v})
		&= \left. \frac{\mathrm{d}}{\mathrm{d} s}\Tr[\mathrm{\Phi}\left(\bm{u}+s\bm{v}\right)]\right|_{s=1}
		- \left. \frac{\mathrm{d}}{\mathrm{d} s}\Tr[\mathrm{\Phi}\left(\bm{u}+s\bm{v}\right)]\right|_{s=0} \\
		&= \int_0^1 \frac{\mathrm{d}}{\mathrm{d} s} \left( \frac{\mathrm{d}}{\mathrm{d} s}\Tr[ \mathrm{\Phi}\left(\bm{u}+s\bm{v}\right)] \right) \mathrm{d} s \\
		&= \int_0^1 {C}_\mathrm{\Phi}(\bm{u}+s\bm{v},\bm{v}) \mathrm{d} s.
		\end{align*}
		Eqs.~\eqref{eq:A_diff} and \eqref{eq:B_diff} can be obtained by Taylor expansion at $(\bm{u,0})$.
		That is,
		\begin{align*}
		&{A}_\mathrm{\Phi} (\bm{u},\epsilon\bm{v}) \\
		&= {A}_\mathrm{\Phi}(\bm{u,0}) 
		+ \mathsf{D}_{\bm{u}} {A}_\mathrm{\Phi}[\bm{u,0}](\bm{0})
		+ \mathsf{D}_{\bm{v}} {A}_\mathrm{\Phi}[\bm{u,0}](\epsilon\bm{v}) \\
		&+ \frac12\left( \mathsf{D}_{\bm{u}}^2 {A}_\mathrm{\Phi}[\bm{u,0}](\bm{0,0})
		+ 2\mathsf{D}_{\bm{u}}\mathsf{D}_{\bm{v}} {A}_\mathrm{\Phi}[\bm{u,0}](\bm{0},\epsilon\bm{v})
		+ \mathsf{D}_{\bm{v}}^2 {A}_\mathrm{\Phi}[\bm{u,0}](\epsilon\bm{v},\epsilon\bm{v}) \right) + o(\epsilon^2) \\
		&= \Tr\left[\mathsf{D}\mathrm{\Phi}[\bm{u+0}](\epsilon\bm{v}) - \mathsf{D}\mathrm{\Phi}[\bm{u}]\left(\mathsf{D}[\bm{v}](\epsilon\bm{v})\right)
		+\frac12 \mathsf{D}^2\mathrm{\Phi}[\bm{u+0}](\epsilon\bm{v},\epsilon\bm{v})\right] + o(\epsilon^2)\\
		&=\frac12 {C}_\mathrm{\Phi}(\bm{u},\bm{v})\epsilon^2 + o(\epsilon^2).
		\end{align*}
		Following the same argument, 
		\begin{align*}
		&{B}_\mathrm{\Phi} (\bm{u},\epsilon\bm{v}) \\
		&= {B}_\mathrm{\Phi}(\bm{u,0}) + \mathsf{D}^2\mathrm{\Phi}[\bm{u+0}](\bm{0},\epsilon\bm{v}) + \mathsf{D}\mathrm{\Phi}[\bm{u+0}]\left( \mathsf{D}[\bm{v}](\epsilon\bm{v}) \right) - \mathsf{D}\mathrm{\Phi}[\bm{u}]\left( \mathsf{D}[\bm{v}](\epsilon\bm{v}) \right)\\
		&+ \frac12 \left( \mathsf{D}^3\mathrm{\Phi}[\bm{u+0}](\bm{0},\epsilon\bm{v},\epsilon\bm{v}) + 2\mathsf{D}^2\mathrm{\Phi}[\bm{u+0}](\epsilon\bm{v},\epsilon\bm{v}) \right) + o(\epsilon^2) \\
		&= {C}_\mathrm{\Phi}(\bm{u},\bm{v})\epsilon^2 + o(\epsilon^2).
		\end{align*}
		We can observe from Eqs.~\eqref{eq:A_int} and \eqref{eq:B_int} that the joint convexity of $(\bm{u},\bm{v})\mapsto {A}_\mathrm{\Phi}(\bm{u},\bm{v})$ and $(\bm{u},\bm{v})\mapsto {B}_\mathrm{\Phi}(\bm{u},\bm{v})$ follows from that of $(\bm{u},\bm{v})\mapsto {C}_\mathrm{\Phi}(\bm{u},\bm{v})$. In other words, we proved that conditions (d)$\Rightarrow$(b) and (d)$\Rightarrow$(c).
		
		Conversely, Eqs.~\eqref{eq:A_diff} and \eqref{eq:B_diff} show that (b)$\Rightarrow$(d) and condition (c)$\Rightarrow$(d). To be more specific, the joint convexity of $(\bm{u},\bm{v})\mapsto {A}_\mathrm{\Phi}(\bm{u},\epsilon\bm{v})$ implies
		\begin{align} \label{eq:A_joint}
		t {A}_\mathrm{\Phi}(\bm{u}_1,\epsilon\bm{v}_1) + (1-t) {A}_\mathrm{\Phi}(\bm{u}_2,\epsilon\bm{v}_2) \geq
		{A}_\mathrm{\Phi}(\bm{u},\epsilon\bm{v}),
		\end{align}
		for each $\bm{u}_1,\bm{u}_2,\bm{v}_1,\bm{v}_2\in\mathbb{M}_d^+$, $t\in[0,1]$, $\epsilon>0$, and $\bm{u}\equiv t\bm{u}_1+(1-t)\bm{u}_2$, $\bm{v}\equiv t\bm{v}_1+(1-t)\bm{v}_2$.
		Invoking Eq.~\eqref{eq:A_diff} gives
		\[
		t {A}_\mathrm{\Phi}(\bm{u}_1,\epsilon\bm{v}_1) + (1-t) {A}_\mathrm{\Phi}(\bm{u}_2,\epsilon\bm{v}_2) = \frac{t{C}_\mathrm{\Phi}(\bm{u}_1,\bm{v}_1) + (1-t) {C}_\mathrm{\Phi}(\bm{u}_2,\bm{v}_2)}{2} \epsilon^2 + o(\epsilon^2),
		\]
		and
		\[
		{A}_\mathrm{\Phi}(\bm{u},\epsilon\bm{v}) = \frac{1}{2}{C}_\mathrm{\Phi}(\bm{u},\epsilon\bm{v})\epsilon^2 + o(\epsilon^2).
		\]
		Hence, Eq.~\eqref{eq:A_joint} is equivalent to
		\[
		t {C}_\mathrm{\Phi}(\bm{u}_1,\bm{v}_1)\epsilon^2 + (1-t) {C}_\mathrm{\Phi}(\bm{u}_2,\bm{v}_2)\epsilon^2 + o(\epsilon^2) \geq {C}_\mathrm{\Phi}(\bm{u},\epsilon\bm{v})\epsilon^2 + o(\epsilon^2).
		\]
		The joint convexity of $(\bm{u},\bm{v})\mapsto {C}_\mathrm{\Phi}(\bm{u},\epsilon\bm{v})$ follows by dividing by $\epsilon^2$ on both sides and letting $\epsilon\rightarrow 0$.
		The joint convexity of $(\bm{u},\bm{v})\mapsto {B}_\mathrm{\Phi}(\bm{u},\epsilon\bm{v})$ can be obtained in a similar way using Eq.~\eqref{eq:B_diff}.

		\item[$(a) \Leftrightarrow (e)$] 
		It is trivial if $\mathrm{\Phi}$ is affine; hence we assume $\mathrm{\Phi} ''>0$.
		We start from the convexity of the map:
		\begin{align} \label{eq:(e)-(f)} 
		\bm{A}\mapsto -\Tr \left[ \bm{h} \left(\mathsf{D}\mathrm{\Psi}[\bm{A}]\right)^{-1} (\bm{h}) \right],\quad \text{for all } \bm{h}\in\mathbb{M}_d^{sa}.
		\end{align}
		To ease the burden of notation, we denote $\mathsf{T}_{\bm{A}} \triangleq  \mathsf{D}\mathrm{\Psi}[\bm{A}] \simeq \mathbb{C}^{d^2\times d^2}$ and $\widehat{\bm{h}} \triangleq  \bm{h} \simeq \mathbb{C}^{d^2 \times 1}$ by the isometric isomorphism between super-operators and matrices.
		Then Eq.~\eqref{eq:(e)-(f)} can be re-written as
		\[
		\bm{A}\mapsto -  \widehat{\bm{h}}^\dagger  \cdot \mathsf{T}_{\bm{A}}^{-1} \cdot \widehat{\bm{h}} ,\quad \text{for all } \widehat{\bm{h}}\in\mathbb{C}^{d^2 \times 1},
		\]
		which is equivalent to the non-negativity of the second derivative (see Proposition \ref{prop:conv2}):
		\begin{align*}
		-\mathsf{D}_{\bm{A}}^2 \left[ \widehat{\bm{h}}^\dagger  \cdot \mathsf{T}_{\bm{A}}^{-1} \cdot \widehat{\bm{h}} \right] (\bm{k},\bm{k}) 
		&= -\widehat{\bm{h}}^\dagger  \cdot \mathsf{D}_{\bm{A}}^2 \left[ \mathsf{T}_{\bm{A}}^{-1} \right] (\bm{k},\bm{k}) \cdot \widehat{\bm{h}}  \\
		&\geq 0,
		\quad \text{for all }\bm{A}\succeq\bm{0},\, \widehat{\bm{h}}\in\mathbb{C}^{d^2 \times 1},\, \bm{k} \in\mathbb{M}_d^\textnormal{sa}.
		\end{align*}

		Now, recall  the chain rule of the Fr\'{e}chet derivative in Proposition \ref{prop:properties}:
		\begin{align*}
		\mathsf{D} \mathcal{F}\circ \mathcal{G} [\bm{A}](\bm{u}) &= \mathsf{D}\mathcal{F}[\mathcal{G}(\bm{A})] \left(\mathsf{D} \mathcal{G} [\bm{A}](\bm{u})\right);\\
		\quad\mathsf{D}^2 \mathcal{F}\circ \mathcal{G} [\bm{A}](\bm{u},\bm{v}) &= \mathsf{D}^2 \mathcal{F}[\mathcal{G}(\bm{A})] \left(\mathsf{D}\mathcal{G}[\bm{A}](\bm{v}),\mathsf{D}\mathcal{G}[\bm{A}](\bm{v})\right)\\
		&+ \mathsf{D}\mathcal{F}[\mathcal{G}(\bm{A})]\left( \mathsf{D}^2 \mathcal{G}[\bm{A}](\bm{u},\bm{v})\right),
		\end{align*}
		and the formula of the differentiation of the inverse function (see Lemma \ref{lemm:2inversion}):
		\begin{align*} 
		\mathsf{D} \mathcal{G}[\bm{A}]^{-1}(\bm{u}) &= - \mathcal{G}(\bm{A})^{-1} \cdot  \mathsf{D}\mathcal{G}[\bm{A}](\bm{u}) \cdot \mathcal{G}(\bm{A})^{-1};\\
		\mathsf{D}^2 \mathcal{G}[\bm{A}]^{-1}(\bm{u},\bm{u}) &= 2 \mathcal{G}(\bm{A})^{-1} \cdot   \mathsf{D}\mathcal{G}[\bm{A}](\bm{u})  \cdot \mathcal{G}(\bm{A})^{-1} \cdot  \mathsf{D}\mathcal{G}[\bm{A}](\bm{u})  \cdot g(\bm{A})^{-1}\\
		&- \mathcal{G}(\bm{A})^{-1} \cdot  \mathsf{D}^2 \mathcal{G}[\bm{A}](\bm{u},\bm{u}) \cdot \mathcal{G}(\bm{A})^{-1}
		,
		\end{align*}
		we can compute the following identities by taking $\mathcal{G}[\bm{A}]\equiv \mathsf{T}_{\bm{A}}$, and $\bm{u}\equiv \bm{k}$:
		\begin{align*}
		\mathsf{D}_{\bm{A}} \left[\mathsf{T}_{\bm{A}}^{-1}\right](\bm{k}) 
		&= - \mathsf{T}_{\bm{A}}^{-1} \cdot \mathsf{D}_{\bm{A}} [\mathsf{T}_{\bm{A}}] (\bm{k}) \cdot \mathsf{T}_{\bm{A}}^{-1},
		\end{align*}
		and
		\begin{align*}
		\mathsf{D}_{\bm{A}} \left[\mathsf{T}_{\bm{A}}^{-1}\right](\bm{k},\bm{k})  &=
		2 \cdot \mathsf{T}_{\bm{A}}^{-1} \cdot \mathsf{D}_{\bm{A}} [\mathsf{T}_{\bm{A}}] (\bm{k}) \cdot \mathsf{T}_{\bm{A}}^{-1}
		\cdot \mathsf{D}_{\bm{A}} [\mathsf{T}_{\bm{A}}](\bm{k}) \cdot \mathsf{T}_{\bm{A}}^{-1}\\
		&- \mathsf{T}_{\bm{A}}^{-1} \cdot \mathsf{D}_{\bm{A}}^2 [\mathsf{T}_{\bm{A}}] (\bm{k},\bm{k}) \cdot \mathsf{T}_{\bm{A}}^{-1}.
		\end{align*}
		
		
		Therefore, we reach the expression (\ref{eq:(f)}), and statement (a) is true if and only if (\ref{eq:(f)}) holds.
		
		Recall that in the scalar case (i.e.~$d=1$), the Fr\'{e}chet derivative can be expressed as the product of the differential and the direction (see e.g.~\cite[Theorem 3.11]{Hig08}):
		\[
		\mathsf{D}\mathrm{\Psi}[a]h = \mathrm{\Psi}'(a) \cdot h.
		\]
		Hence, Eq.~\eqref{eq:(f)} reduces to
		\begin{align*}
		& h \cdot \left(\mathrm{\Psi}'(a)\right)^{-1} \cdot \mathrm{\Psi} ''' (a) \cdot k^2\cdot \left(\mathrm{\Psi}'(a)\right)^{-1} \cdot h \\
		&= \frac{\mathrm{\Phi} ''''(a) \cdot \mathrm{\Phi} '' (a) \cdot k^2 h^2}{\mathrm{\Phi} '' (a)^2}\\
		&\geq 2 \cdot h \cdot \left(\mathrm{\Psi}'(a)\right)^{-1} \cdot \mathrm{\Psi} ''(a) \cdot k \cdot \left(\mathrm{\Psi}'(a)\right)^{-1} \cdot \mathrm{\Psi} '' (a) \cdot k \cdot \left(\mathrm{\Psi}'(a)\right)^{-1} \cdot h \\
		&= \frac{ 2\mathrm{\Phi} '''(a) ^2 \cdot k^2 h^2}{\mathrm{\Phi} ''(a)^3}.
		\end{align*}
		for all $a>0$ and $h,k\in\mathbb{R}$.
		In other words, Eq.~\eqref{eq:(f)} can be viewed as a non-commutative generalisation of the classical statement: $\mathrm{\Phi} '''' \mathrm{\Phi} '' \geq 2 \mathrm{\Phi}'''^2$.

		\item[$(d)\Leftrightarrow (f)$]
		
		For any $t\in[0,1]$, define $\bm{F}_t:\mathbb{M}_d^+ \times \mathbb{M}_d^+ \rightarrow \mathbb{M}_d^{sa}$ as
		\[
		\bm{F}_t(\bm{X},\bm{Y})\triangleq  t\mathrm{\Phi}(\bm{X}) + (1-t) \mathrm{\Phi}(\bm{Y}) - \mathrm{\Phi}(t\bm{X}+(1-t)\bm{Y}).
		\]
		By taking $x\equiv(\bm{X},\bm{Y})$ and $h\equiv (\bm{h},\bm{k})$ in Proposition~\ref{prop:conv2},  the convexity of the twice Fr\'{e}chet differentiable function $\bm{F}_t$ is equivalent to 
		\[
		\mathsf{D}^2 \bm{F}_t[\bm{X},\bm{Y}](\bm{h},\bm{k})\succeq \bm{0} \quad \forall \bm{X},\bm{Y}\in\mathbb{M}_d^+ \quad \text{and} \quad \forall\bm{h},\bm{k}\in\mathbb{M}_d^{sa}.
		\]
		Then,  with the help of \emph{partial Fr\'{e}chet derivative} defined in Proposition \ref{Prop_pfd}, the second-order Fr\'{e}chet derivative of $\bm{F}_t(\bm{X},\bm{Y})$ can be evaluated as
		\begin{align}
		&\mathsf{D}^2 \bm{F}_t[\bm{X},\bm{Y}](\bm{h},\bm{k})\notag \\
		&= \mathsf{D}^2_{\bm{X}} \bm{F}_t[\bm{X},\bm{Y}](\bm{h},\bm{h})
		+ \mathsf{D}_{\bm{Y}} \mathsf{D}_{\bm{X}} \bm{F}_t[\bm{X},\bm{Y}](\bm{h},\bm{k})\notag \\
		& \qquad + \mathsf{D}_{\bm{X}} \mathsf{D}_{\bm{Y}} \bm{F}_t[\bm{X},\bm{Y}](\bm{k},\bm{h})
		+ \mathsf{D}^2_{\bm{Y}} \bm{F}_t[\bm{X},\bm{Y}](\bm{k},\bm{k})\notag \\
		&= t \cdot\mathsf{D}^2 \mathrm{\Phi}[\bm{X}] (\bm{h},\bm{h}) - t^2 \cdot \mathsf{D}^2 \mathrm{\Phi}[t\bm{X}+(1-t)\bm{Y}](\bm{h},\bm{h})\notag\\
		& -t(1-t) \cdot \mathsf{D}^2 \mathrm{\Phi}\left[t\bm{X}+(1-t)\bm{Y}\right] (\bm{h},\bm{k}) 
		-t(1-t) \cdot \mathsf{D}^2 \mathrm{\Phi}\left[t\bm{X}+(1-t)\bm{Y}\right] (\bm{k},\bm{h})\notag \\
		&\quad+ (1-t) \cdot\mathsf{D}^2 \mathrm{\Phi}[\bm{Y}] (\bm{k},\bm{k}) - (1-t)^2 \cdot\mathsf{D}^2 \mathrm{\Phi}[t\bm{X}+(1-t)\bm{Y}](\bm{k},\bm{k}). \label{eq_no10}
		\end{align}
		Taking trace on both sides of (\ref{eq_no10}) and invoking Lemma~\ref{lemm:trace_Petz},  we have
		\begin{align}
		&\Tr \left[ 	\mathsf{D}^2 \bm{F}_t[\bm{X},\bm{Y}](\bm{h},\bm{k}) \right] \notag\\
		&= \Tr \left[ t \cdot \bm{h}\mathsf{D} \mathrm{\Psi}[\bm{X}] (\bm{h}) - t^2 \cdot \bm{h} \mathsf{D} \mathrm{\Psi}[t\bm{X}+(1-t)\bm{Y}](\bm{h}) \right] \notag\\
		&- \Tr \left[ t(1-t) \cdot \bm{h}\mathsf{D} \mathrm{\Psi}\left[t\bm{X}+(1-t)\bm{Y}\right] (\bm{k}) 
		+ t(1-t) \cdot \bm{k}\mathsf{D} \mathrm{\Psi}\left[t\bm{X}+(1-t)\bm{Y}\right] (\bm{h})
		\right] \notag\\
		&+ \Tr \left[ (1-t) \cdot \bm{k} \mathsf{D} \mathrm{\Psi}[\bm{Y}] (\bm{k}) - (1-t)^2 \cdot \bm{k} \mathsf{D}^2 \mathrm{\Psi}[t\bm{X}+(1-t)\bm{Y}](\bm{k}) \right]. \label{eq:trace_double}
		\end{align}
		Since both the trace and the second-order Fr\'{e}chet derivative are bilinear, we have the following result
		\begin{align}
		&\Tr \left[ t^2 \cdot \bm{h} \mathsf{D} \mathrm{\Psi}[t\bm{X}+(1-t)\bm{Y}](\bm{h}) + t(1-t) \cdot \bm{k} \mathsf{D} \mathrm{\Psi}[t\bm{X}+(1-t)\bm{Y}](\bm{h}) \right]\notag\\
		&= \left\langle t\bm{h} , \mathsf{D} \mathrm{\Psi}[t\bm{X}+(1-t)\bm{Y}](t\bm{h}) \right\rangle
		+ \left\langle (1-t)\bm{k}, \mathsf{D} \mathrm{\Psi}[t\bm{X}+(1-t)\bm{Y}](t\bm{h}) \right\rangle \notag \\
		&=\left\langle t\bm{h}+(1-t)\bm{k} , \mathsf{D} \mathrm{\Psi}[t\bm{X}+(1-t)\bm{Y}](t\bm{h}) \right\rangle. \label{eq:temp1}
		\end{align}
		Similarly, 
		\begin{align}
		&\Tr \left[ t(1-t) \cdot \bm{h} \mathsf{D} \mathrm{\Psi}[t\bm{X}+(1-t)\bm{Y}](\bm{k}) 
		+ (1-t)^2 \cdot \bm{k} \mathsf{D} \mathrm{\Psi}[t\bm{X}+(1-t)\bm{Y}](\bm{k}) \right]\notag\\
		&=\left\langle t\bm{h}+(1-t)\bm{k} , \mathsf{D} \mathrm{\Psi}[t\bm{X}+(1-t)\bm{Y}]((1-t)\bm{k}) \right\rangle. \label{eq:temp2}
		\end{align}
		Combining Eqs.~(\ref{eq:temp1}) and (\ref{eq:temp2}), Eq.~(\ref{eq:trace_double}) can be expressed as
		\begin{align*}
		\Tr \left[ 	\mathsf{D}^2 \bm{F}_t[\bm{X},\bm{Y}](\bm{h},\bm{k}) \right]
		&= t \cdot \left\langle \bm{h}, \mathsf{D} \mathrm{\Psi}[\bm{X}] (\bm{h}) \right\rangle +
		(1-t) \cdot \left\langle \bm{k}, \mathsf{D} \mathrm{\Psi}[\bm{Y}] (\bm{k}) \right\rangle \\
		&- \left\langle (t\bm{h}+(1-t)\bm{k}),\mathsf{D}\mathrm{\Psi}[t\bm{X}+(1-t)\bm{Y}](t\bm{h}+(1-t)\bm{k}) \right\rangle.
		\end{align*}
		Then, it is not hard to observe that the non-negativity of $	\Tr \left[ 	\mathsf{D}^2 \bm{F}_t[\bm{X},\bm{Y}](\bm{h},\bm{k}) \right]$ for every $\bm{X},\bm{Y}\in\mathbb{M}_d^{+}$,  $\bm{h},\bm{k}\in\mathbb{M}_d^\textnormal{sa}$, and $t\in[0,1]$ is equivalent to the joint convexity of the map
		\[
		(\bm{X},\bm{A})\mapsto \left\langle \bm{X},  \mathsf{D}\mathrm{\Psi} [\bm{A}] (\bm{X}) \right\rangle = \Tr\left[\mathsf{D}^2 \mathrm{\Phi}[\bm{A}](\bm{X},\bm{X}) \right].
		\]

		\item[$(j) \Rightarrow (g)$]
		
		Considering $n=2$, the sub-additivity means that
		\[
		H_\mathrm{\Phi}(\bm{Z}) \leq \mathbb{E}_1 H_\mathrm{\Phi}^{(2)}(\bm{Z}) + \mathbb{E}_2 H_\mathrm{\Phi}^{(1)}(\bm{Z}).
		\]
		Then, we have
		\begin{align*}
		\mathbb{E}_1 H_\mathrm{\Phi}^{(2)}(\bm{Z})
		&\geq H_\mathrm{\Phi}(\bm{Z}) - \mathbb{E}_2 H_\mathrm{\Phi}^{(1)} (\bm{Z}) \\
		&= 	\mathbb{E} \mathrm{\Phi}(\bm{Z}) - \mathrm{\Phi}(\mathbb{E}\bm{Z}) - \mathbb{E}_2\mathbb{E}_1 \mathrm{\Phi}(\bm{Z}) + \mathbb{E}_2 \mathrm{\Phi}(\mathbb{E}_1\bm{Z}) \\
		&=\mathbb{E}_2 \mathrm{\Phi}(\mathbb{E}_1 \bm{Z}) - \mathrm{\Phi}(\mathbb{E}_2 \mathbb{E}_1 \bm{Z}) \\
		&= H_\mathrm{\Phi}(\mathbb{E}_1 \bm{Z}).
		\end{align*}
		
		\item[$(f)\Leftrightarrow (h)$]
		
		Let $s\in[0,1]$, define a pair of positive semi-definite random matrices $(\bm{X},\bm{Y})$ taking values $(\bm{x,y})$ with probability $s$ and $(\bm{x',y'})$ with probability $(1-s)$.
		Then the convexity of $H_\mathrm{\Phi}$ implies that
		\begin{align} \label{eq:H_convex}
		H_\mathrm{\Phi}(t\bm{X}+(1-t)\bm{Y}) \leq t H_\mathrm{\Phi}(\bm{X}) + (1-t) H_\mathrm{\Phi}(\bm{Y})
		\end{align}
		for every $t\in[0,1]$.
		Now define $F_t(\bm{u},\bm{v}):\mathbb{M}_d^+\times\mathbb{M}_d^+\rightarrow \mathbb{R}$ as
		\[
		F_t(\bm{u},\bm{v})\triangleq \Tr\left[ t\mathrm{\Phi}(\bm{u}) + (1-t) \mathrm{\Phi}(\bm{v}) - \mathrm{\Phi}(t\bm{u}+(1-t)\bm{v})\right].
		\]
		Then, it follows that
		\begin{align*}
		&  s F_t(\bm{x,y}) + (1-s) F_t(\bm{x',y'}) - F_t( s(\bm{x,y}) + (1-s) \bm{x',y'} )\\
		&= t \,\mathbb{E}\mathrm{\Phi}(\bm{X}) - t\mathrm{\Phi}(\mathbb{E}\bm{X}) + (1-t) \mathbb{E} \mathrm{\Phi}(\bm{Y}) - (1-t) \mathrm{\Phi}(\mathbb{E}\bm{Y}) \\
		&- \mathbb{E}\mathrm{\Phi}\left(t\bm{X}+(1-t)\bm{Y}\right) + \mathrm{\Phi}\left(t\mathbb{E}\bm{X} + (1-t)\mathbb{E} \bm{Y}\right)\\
		&= t H_\mathrm{\Phi}(\bm{X}) + (1-t) H_\mathrm{\Phi}(\bm{Y}) - H_\mathrm{\Phi}(t \bm{X} + (1-t)\bm{Y}),
		\end{align*}
		which means that the convexity of the pair $(\bm{u},\bm{v})\mapsto F_t(\bm{u},\bm{v})$ is equivalent to the convexity of $H_\mathrm{\Phi}$, i.e.~Eq.~\eqref{eq:H_convex}.
		
		\item[$(g)\Leftrightarrow (h)$]
		
		
		Define a positive semi-definite random matrix $\bm{Z}\triangleq f(\bm{X}_1,\bm{X}_2)$, which depends on two random variables $\bm{X}_1,\bm{X}_2$ on a Polish space.
		Denote by $\bm{Z}_{\bm{X}_1}$ the random matrix $\bm{Z}$ conditioned on $\bm{X}_1$.
		According to the convexity of ${H}_\mathrm{\Phi}$, it follows that
		\begin{align*}
		\mathbb{E}_1{H}_\mathrm{\Phi}(\bm{Z}|\bm{X}_1) 
		&=  \mathbb{E}_1{H}_\mathrm{\Phi}(\bm{Z}_{\bm{X}_1}) \\
		&= \mathbb{E}_1 \Big[ \tr \big( \mathbb{E}_2 \mathrm{\Phi}(\bm{Z}_{\bm{X}_1}) - \mathrm{\Phi}(\mathbb{E}_2 \bm{Z}_{\bm{X}_1}) \big) \Big] \\
		&\geq   \tr \mathbb{E}_2 \mathrm{\Phi}(\mathbb{E}_1 \bm{Z}_{\bm{X}_1}) - \tr \left[ \mathrm{\Phi}(\mathbb{E}_1 \mathbb{E}_2 \bm{Z}_{\bm{X}_1}) \right] \\
		&= {H}_\mathrm{\Phi}(\mathbb{E}_1 \bm{Z}).
		\end{align*}

		Conversely, define a positive semi-definite random matrix $\bm{Z}(s,\bm{X},\bm{Y}) \triangleq  s\bm{X} + (1-s)\bm{Y}$ where $s$ is a random variable.
		Now let $s$ be Bernoulli distributed with parameter $t \in [0,1]$.
		Then for all $t\in[0,1]$, the inequality $\mathbb{E}_1 H_\mathrm{\Phi}(\bm{Z}|s)\geq H_\mathrm{\Phi}(\mathbb{E}_1 \bm{Z})$ coincides
		\[
		H_\mathrm{\Phi}(t\bm{X}+(1-t)\bm{Y}) \leq t H_\mathrm{\Phi}(\bm{X}) + (1-t) H_\mathrm{\Phi}(\bm{Y}).
		\]
		
	\end{description}	
\end{proof}

\section{Proof of Theorem {\ref{theo:sub_oper}}} \label{proof:sub_oper}
Our approach of proving operator subadditivity (Theorem \ref{theo:sub_oper}) parallels \cite[Theorem 2.5]{CT14} and \cite[Section 3.1]{BBM+05}. 
The strategy is as the following.  First, we prove the supremum representation for the operator-valued $\Phi$-entropies in Section \ref{Rep}.
Second, we establish a conditional operator Jensen's inequality in Section \ref{Jensen}. Finally, we arrive at the proof of Theorem \ref{theo:sub_oper} in Section \ref{proof2}.

\subsection{Representation of operator-valued $\mathrm{\Phi}$-entropy} \label{Rep}

\begin{theo}[Supremum Representation for Operator-Valued $\mathrm{\Phi}$-Entropies] \label{theo:representation}
	Fix a function $\mathrm{\Phi}\in \textnormal{\textbf{(C3)}}$. Assume $\bm{Z}\in\mathbb{M}_d^{++}$ is a random positive definite matrix for which $|\bm{Z}|$, $|\mathrm{\Phi}(\bm{Z})\|$ are Bochner integrable.
	Then the operator-valued $\mathrm{\Phi}$-entropy can be represented as
	\begin{eqnarray} \label{eq:oper_represent}
	\bm{H}_\mathrm{\Phi}(\bm{Z})=\sup_{\bm{T}\succ \bm{0}}\mathbb{E}\left[\mathsf{D}\mathrm{\Phi}[\bm{T}](\bm{Z-T})-\mathsf{D}\mathrm{\Phi}[\,\mathbb{E}\bm{T}](\bm{Z-T})+\mathrm{\Phi}(\bm{T})-\mathrm{\Phi}(\mathbb{E}\bm{T})\right].
	\end{eqnarray}
	The range of the supremum contains each random positive definite matrix $\bm{T}$ for which $|\bm{T}|$ and $|\Phi(\bm{T})|$ are Bochner integrable.
	In particular, the normalised matrix $\mathrm{\Phi}$-entropy can be written in the dual form
	\begin{eqnarray} \label{eq:dual2}
	\bm{H}_{\mathrm{\Phi}} (\bm{Z})=\sup_{\bm{T}\succ \bm{0}} \mathbb{E} \ \left[ \bm{\Upsilon}_1(\bm{T}, \bm{Z})+\bm{\Upsilon}_2(\bm{T})\right],
	\end{eqnarray}
	where $\bm{\Upsilon}_1(\bm{T},\bm{Z})=\mathsf{D}\mathrm{\Phi}[\bm{T}](\bm{Z})-\mathsf{D}\mathrm{\Phi}[\mathbb{E}\,\bm{T}](\bm{Z})$ is a linear map of $\bm{Z}$ and $\bm{\Upsilon}_2(\bm{T})=-\mathsf{D}\mathrm{\Phi}[\bm{T}](\bm{T})+\mathsf{D}\mathrm{\Phi}[\mathbb{E}\bm{T}](\bm{T})+ \left(\mathrm{\Phi}(\bm{T})-\mathrm{\Phi}(\mathbb{E}\,\bm{T})\right)$.
\end{theo}
\begin{proof}
	Observe that when $\bm{T}=\bm{Z}$, the right-hand side of Eq.~\eqref{eq:oper_represent} equals $\bm{H}_{\mathrm{\Phi}} (\bm{Z})$. Then it remains to confirm the inequality
	\begin{eqnarray} \label{eq:dual3}
	\bm{H}_{\mathrm{\Phi}} (\bm{Z})\succeq \mathbb{E}\left[\mathsf{D}\mathrm{\Phi}[\bm{T}](\bm{Z-T})-\mathsf{D}\mathrm{\Phi}[\,\mathbb{E}\bm{T}](\bm{Z-T})+\mathrm{\Phi}(\bm{T})-\mathrm{\Phi}(\mathbb{E}\bm{T})\right]
	\end{eqnarray}
	for each random positive definite matrix $\bm{T}$ that satisfies the integrability conditions. We follow the interpolation argument as in \cite[Lemma 4.1]{CT14}.
	For $s\in[0,1]$, define the matrix-valued function
	\[
	\bm{F}(s)=\mathbb{E}\left[\mathsf{D}\mathrm{\Phi}[\bm{T}_s](\bm{Z}-\bm{T}_s)-\mathsf{D}\mathrm{\Phi}[\mathbb{E}\,\bm{T}_s](\bm{Z}-\bm{T}_s)\right]+\bm{H}_\mathrm{\Phi}(\bm{T}_s).
	\]
	where
	\[
	\bm{T}_s\triangleq(1-s)\cdot \bm{Z}+s\cdot\bm{T}\quad \text{for } s\in[0,1].
	\]
	Note that $\bm{F}(0)=H_{\mathrm{\Phi}}(\bm{Z})$, and $\bm{F}(1)$ matches the right-hand side of Eq.~\eqref{eq:dual3}.
	As a result, it suffices to show that $\bm{F}'(s)\leq0$ for $s\in[0,1]$ in order to verify Eq.~\eqref{eq:dual3}.
	By the replacement $\bm{Z}-\bm{T}_s=-s\cdot(\bm{T}-\bm{Z})$, the function $\bm{F}(s)$ can be rephrased as
	\[ \label{eq:dual4}
	\bm{F}(s)=-s\cdot\mathbb{E}\left[\mathsf{D}\mathrm{\Phi}[\bm{T}_s](\bm{T-Z})-\mathsf{D}\mathrm{\Phi}[\mathbb{E}\bm{T}_s](\bm{T-Z})\right]+\mathbb{E}\left[\mathrm{\Phi}(\bm{T}_s)-\mathrm{\Phi}(\mathbb{E}\,\bm{T}_s)\right].
	\]

	Differentiate the above function to arrive at
	\begin{subequations}
		\begin{align} 
		\bm{F}'(s) = &-s\,\mathbb{E}\left[\mathsf{D}^2\mathrm{\Phi}[\bm{T}_s](\bm{T-Z},\bm{T-Z})\right]+s\,\mathbb{E}\left[\mathsf{D}^2\mathrm{\Phi}[\mathbb{E}\bm{T}_s](\bm{T-Z},\mathbb{E}(\bm{T-Z}))\right] \notag \\ 
		&-\mathbb{E}\left[\mathsf{D}\mathrm{\Phi}[\bm{T}_s](\bm{T-Z})-\mathsf{D}\mathrm{\Phi}[\mathbb{E}\bm{T}_s](\bm{T-Z})\right]+\mathbb{E}\left[\mathsf{D}\mathrm{\Phi}[\bm{T}_s](\bm{T-Z})-\mathsf{D}\mathrm{\Phi}[\mathbb{E}\bm{T}_s](\bm{T-Z})\right] \label{eq:dual5} \\ 
		= &-s\,\mathbb{E}\left[\mathsf{D}^2\mathrm{\Phi}[
		{T}_s](\bm{T-Z},\bm{T-Z})+s\mathsf{D}^2\mathrm{\Phi}[\mathbb{E}\bm{T}_s](\mathbb{E}(\bm{T-Z}),\mathbb{E}(\bm{T-Z}))\right], \label{eq:dual6}
		\end{align}
	\end{subequations}	
	where we cancel the last two terms in Eq.~\eqref{eq:dual5} and the second equation \eqref{eq:dual6} follows from the bilinearity of the second order Fr\'{e}chet differentiation.
	
	Invoke the joint convexity condition of the function $\mathsf{D}^2\mathrm{\Phi}[\bm{T}_s](\bm{T-Z},\bm{T-Z})$ (see Eq.~\eqref{eq:entropy_class}), we establish the above derivative to be negative semi-definite, i.e.~$\bm{F}'(s)\preceq\bm{0}$ for $s\in[0,1]$ and thus complete the proof.
\end{proof}

\subsection{A conditional operator Jensen's inequality} \label{Jensen}

\begin{lemm}[Conditional Operator Jensen's Inequality for Operator-Valued $\mathrm{\Phi}$-Entropy]\label{lemm:Jensen}
	Suppose that $(\bm{X}_1,\bm{X}_2)$ is a pair of independent random matrices taking values in a Polish space, and let $\bm{Z}=\bm{Z}(\bm{X}_1,\bm{X}_2)$ be a positive definite random matrix for which $|\bm{Z}|$ and $|\Phi(\bm{Z})|$ are Bochner integrable. Then
	\[
	\bm{H}_{\mathrm{\Phi}} (\mathbb{E}_1 \bm{Z})\preceq \mathbb{E} \bm{H}_{\mathrm{\Phi}}(\bm{Z}|\bm{X}_1),
	\]
	where $\mathbb{E}_1$ is the expectation with respect to the first matrix $\bm{X}_1$.
\end{lemm}
\begin{proof}
	Let $\mathbb{E}_2$ refer to the expectation with respect to the second matrix $\bm{X}_2$.
	In the following, we use $\bm{T}(\bm{X}_2)$ to emphasise the matrix $\bm{T}$ depends only on the randomness in $\bm{X}_2$. Recall the supremum representation, Eq.~\eqref{eq:dual2}, we have:
	\begin{align*}
	\bm{H}_{\mathrm{\Phi}}(\mathbb{E}_1 \bm{Z}) &= \sup_{\bm{T}} \mathbb{E}_2 \left[ \bm{\Upsilon}_1\left(\bm{T}(\bm{X}_2),\mathbb{E}_1 \bm{Z}\right)+\bm{\Upsilon}_2\left(\bm{T}(\bm{X}_2)\right)\right]\\
	&= \sup_{\bm{T}} \mathbb{E}_1\mathbb{E}_2 \left[ \bm{\Upsilon}_1\left(\bm{T}(\bm{X}_2), \bm{Z}\right)+\bm{\Upsilon}_2\left(\bm{T}(\bm{X}_2)\right)\right]\\
	&\preceq \mathbb{E}_1 \sup_{\bm{T}} \mathbb{E}_2 \left[ \bm{\Upsilon}_1\left(\bm{T}(\bm{X}_2),\bm{Z}\right) +\bm{\Upsilon}_2\left(\bm{T}\bm{X}_2)\right)\right]\\
	&= \mathbb{E}_1 \sup_{\bm{T}} \mathbb{E}\left[  \bm{\Upsilon}_1\left(\bm{T}(\bm{X}_2),\bm{Z}\right)+\bm{\Upsilon}_2\left(\bm{T}(\bm{X}_2)\right)\big|\bm{X}_1\right]\\
	&= \mathbb{E}_1 \bm{H}_{\mathrm{\Phi}}(\bm{Z}|\bm{X}_1).
	\end{align*}
	The second relation follows from the Fubini's theorem to interchange the order of $\mathbb{E}_1$ and $\mathbb{E}_2$.
	In the third line we use the convexity of the supremum. (Note that it is not always true under partial ordering.
	However, it holds in our case because the supremum is attained when $\bm{T}\equiv\mathbb{E}_1 \bm{Z}$ in the second line.)
	The last identity is exactly the supremum representation Eq.~\eqref{eq:dual2} in the conditional form.
\end{proof}

It is worth emphasising that the conditional Jensen inequality can also be achieved by item (d) in Theorem \ref{theo:operator} (cf.~(f)$\Leftrightarrow$(g)$\Leftrightarrow$(h) in Theorem \ref{theo:trace}).

\subsection{Subadditivity of operator-valued $\Phi$-entropies} \label{proof2}

Now we are at the position to prove the subadditivity of the operator-valued $\mathrm{\Phi}$-entropies.

\begin{proof}
	By adding and subtracting the term $\mathrm{\Phi}(\mathbb{E}_1 \bm{Z})$, the operator-valued $\mathrm{\Phi}$-entropy can be expressed as
	\begin{align} \label{eq:mod_entropy2}
	\bm{H}_{\mathrm{\Phi}}(\bm{Z}) &= \mathbb{E}  \left[\mathrm{\Phi}(\bm{Z})-\mathrm{\Phi}(\mathbb{E}_1 \bm{Z})+\mathrm{\Phi}(\mathbb{E}_1 \bm{Z})-\mathrm{\Phi}(\mathbb{E} \bm{Z})\right]\notag\\
	&= \mathbb{E}  \left[\mathbb{E}_1\mathrm{\Phi}(\bm{Z})-\mathrm{\Phi}(\mathbb{E}_1 \bm{Z})\right]
	+\left[\mathbb{E}\mathrm{\Phi}(\mathbb{E}_1\bm{Z})-\mathrm{\Phi}(\mathbb{E}\,\mathbb{E}_1 \bm{Z})\right]\notag\\
	&= \mathbb{E} \bm{H}_{\mathrm{\Phi}}(\bm{Z}|\bm{X}_{-1})+\bm{H}_{\mathrm{\Phi}}(\mathbb{E}_1 \bm{Z})\notag\\
	&\preceq \mathbb{E}\bm{H}_{\mathrm{\Phi}}(\bm{Z}|\bm{X}_{-1})+ \mathbb{E}_1 \bm{H}_{\mathrm{\Phi}}(\bm{Z}|\bm{X}_1),
	\end{align}
	where the last inequality results from Lemma \ref{lemm:Jensen} since $\bm{X}_1$ is independent from $\bm{X}_{-1}$.
	
	Following the same reasoning we obtain the operator-valued $\Phi$-entropy conditioned on $\bm{X}_1$:
	\[ \label{eq:mod_entropy3}
	\bm{H}_{\mathrm{\Phi}}(\bm{Z}|\bm{X}_1)\preceq \mathbb{E}\left[\bm{H}_{\mathrm{\Phi}}(\bm{Z}|\bm{X}_{-2})\big|\bm{X}_1\right]+ \mathbb{E}_2 \bm{H}_{\mathrm{\Phi}}(\bm{Z}|\bm{X}_1,\bm{X}_2).
	\]
	By plugging the expression into Eq.~\eqref{eq:mod_entropy2} we get
	\[ \label{eq:mod_entropy4}
	\bm{H}_{\mathrm{\Phi}}(\bm{Z})\preceq \sum_{i=1}^2\mathbb{E} \bm{H}_{\mathrm{\Phi}}(\bm{Z}|\bm{X}_{-i})+\mathbb{E}_1\mathbb{E}_2 \bm{H}_{\mathrm{\Phi}} (\bm{Z}|\bm{X}_1,\bm{X}_2).
	\]
	Finally, by repeating this procedure we achieve the subadditivity of the operator-valued $\mathrm{\Phi}$-entropy
	\[
	\bm{H}_{\mathrm{\Phi}}(\bm{Z}) \preceq \sum_{i=1}^n \mathbb{E} \left[ \bm{H}_{\mathrm{\Phi}}(\bm{Z}|\bm{X}_{-i})\right],
	\]
	which completes our claim.
\end{proof}

\section{Conclusion}
In this paper, we extend results of Chen and Tropp \cite{CT14}, Pitrik and Virosztek \cite{PV14}, and Hansen and Zhang \cite{HZ14} to complete the characterisations of the matrix $\Phi$-entropy functionals.
Moreover, we generalise the matrix $\Phi$-entropy functionals to the operator-valued $\Phi$-entropies, and show that this generalisation preserves the subadditivity property. Additionally, we prove that the set of operator-valued $\Phi$-entropies is not empty and contains at least the square function. Equivalent characterisations of the operator-valued $\Phi$-entropies are also derived. 
This result demonstrates that the subadditivity of $\bm{H}_\Phi(\bm{Z})$ is equivalent to the operator convexity of $\bm{H}_\Phi(\bm{Z})$ on the convex cone of $\bm{Z}$. 
Finally, we exploit the subadditivity to prove the operator Efron-Stein inequality. It is promising that the proposed result can also derive the \emph{matrix exponential Efron-Stein} (cf.~ \cite[Theorem 4.3]{PMT14}) and the moment inequalities for random matrices; see \cite{BBM+05} and \cite[Chapter 15]{BLM13}.

The subadditivity of matrix $\Phi$-entropies leads to a series of important inequalities: matrix Poincar\'e inequalities with respect to binomial and Gaussian distributions, and the related matrix logarithmic Sobolev inequalities \cite{CH1}. 
In Ref.~\cite{CH2}, the subadditivity and the operator Efron-Stein inequality can be exploited to estimate the mixing time of a quantum random graph. It enables us to better understand the dynamics and long-term behaviours of a quantum system undergoing Markovian processes. We believe the proposed results will lead to more matrix functional inequalities, and have substantial impact in operator algebra and quantum information science. 

Finally, we remark that the results of operator-valued $\Phi$-entropies and the operator Efron-Stein inequalities hold in the infinite-dimensional setting. This is not hard to verify because the tools (such as Fr\'echet derivatives) employed in the proofs hold in the infinite dimension.  










\section*{Acknowledgement}
The authors thank Marco Tomamichel for helpful discussion about the operator-valued $\Phi$-entropies. MH is supported by an ARC Future Fellowship under Grant FT140100574.

\appendix

\section{Miscellaneous Lemmas} \label{lemmas}

\begin{prop}[Properties of Fr\'{e}chet Derivatives {\cite[Theorem 3.4]{Hig08}}] \label{prop:properties}
	Let $\mathcal{U},\mathcal{V}$ and $\mathcal{W}$ be real Banach spaces. Let $\mathcal{L}_1:\mathcal{U}\rightarrow\mathcal{V}$ and $\mathcal{L}_2:\mathcal{V}\rightarrow\mathcal{W}$ be Fr\'{e}chet differentiable at $\bm{A}\in\mathcal{U}$ and $\mathcal{L}_1(\bm{A})$ respectively, and let $\mathcal{L} = \mathcal{L}_2 \circ \mathcal{L}_1$ (i.e.~$\mathcal{L}(\bm{A}) = \mathcal{L}_2\left( \mathcal{L}_1 (\bm{A}) \right)$. Then $\mathcal{L}$ is Fr\'{e}chet differentiable at $\bm{A}$ and $\mathsf{D}\mathcal{L}[\bm{A}](\bm{E}) = \mathsf{D}\mathcal{L}_2 [\mathcal{L}_1(\bm{A})] \left( \mathsf{D}\mathcal{L}_1[\bm{A}](\bm{E}) \right)$.
\end{prop}


\begin{prop}[Convexity of twice Fr\'echet differentiable matrix functions {\cite[Proposition 2.2]{Han97}}] \label{prop:conv2}
	Let $U$ be an open convex subset of a real Banach space $\mathcal{U}$, and $\mathcal{W}$ is also a real Banach space. Then a twice Fr\'echet differentiable function $\mathcal{L}:U \to \mathcal{W}$ is convex if and only if $\mathsf{D}^2 \mathcal{L}(\bm{X})(\bm{h},\bm{h})\succeq \bm{0}$ for each $\bm{X}\in U$ and $\bm{h}\in \mathcal{U}$.
\end{prop}


\begin{prop}[Partial Fr\'{e}chet derivative {\cite[Proposition 5.3.15]{AH09}}] \label{Prop_pfd}
	If $\mathcal{L}:\mathcal{U}\times \mathcal{V} \rightarrow \mathcal{W}$ is Fr\'echet differentiable at $(\bm{X},\bm{Y})\in \mathcal{U}\times \mathcal{V}$, then the partial Fr\'{e}chet derivatives $\mathsf{D}_{\bm{X}} \mathcal{L}[\bm{X},\bm{Y}]$ and $\mathsf{D}_{\bm{Y}} \mathcal{L}[\bm{X},\bm{Y}]$ exist, and
	\[
	\mathsf{D}\mathcal{L}[\bm{X},\bm{Y}](\bm{h},\bm{k}) = \mathsf{D}_{\bm{X}} \mathcal{L}[\bm{X},\bm{Y}] (\bm{h}) + \mathsf{D}_{\bm{Y}} \mathcal{L}[\bm{X},\bm{Y}] (\bm{k}).
	\]
\end{prop}

\begin{prop}[{\cite[Theorem 2.2]{HP95}}] \label{Prop_trace_Petz}
	Let $ \bm{A}, \bm{X}\in\mathbb{M}^{sa}$ and $t\in\mathbb{R}$. Assume $f:I\to \mathbb{R}$ is a continuously differentiable function defined on interval $I$ and assume that the eigenvalues of $ \bm{A}+t\bm{X} \subset I$. Then
	\[
	\left.\frac{\mathrm{d} }{ \mathrm{d}t} \Tr f(\bm{A}+t \bm{X})\right|_{t=t_0} = \Tr [  \bm{X} f' ( \bm{A} + t_0 \bm{X}) ].
	\]
\end{prop}
Proposition~\ref{Prop_trace_Petz} directly leads to the following lemma.
\begin{lemm} \label{lemm:trace_Petz}
	Let $ \bm{A}, \bm{X}, \bm{Y}\in\mathbb{M}^{sa}$ and $t\in\mathbb{R}$. Assume $f:I\to \mathbb{R}$ is a continuously differentiable function defined on interval $I$, and assume that the eigenvalues of $ \bm{A}+t\bm{X} \subset I$. Then
	\begin{align*}
	\Tr ( \mathsf{D}^2 f [ \bm{A} ] ( \bm{X} , \bm{Y} ) )  
	&= 	\left\langle \bm{X}, \mathsf{D} f' [ \bm{A} ] ( \bm{Y} ) \right\rangle 
	= \left\langle \bm{Y}, \mathsf{D} f' [ \bm{A} ] ( \bm{X} ) \right\rangle.
	\end{align*}
\end{lemm}

\begin{lemm}[Second-Order Fr\'{e}chet Derivative of Inversion Function] \label{lemm:2inversion}
	Let $\mathcal{G}:\mathbb{M}\rightarrow \mathbb{M}$ be second-order Fr\'{e}chet differentiable at $\bm{A}\in\mathbb{M}$, and $\mathcal{G}(\bm{A})$ be invertible.
	Then, for each $\bm{h}, \bm{k} \in\mathbb{M}$, we have
	\begin{align*}
	\mathsf{D} \mathcal{G}[\bm{A}]^{-1}(\bm{h}) &= - \mathcal{G} (\bm{A})^{-1} \cdot  \mathsf{D} \mathcal{G} [\bm{A}](\bm{h}) \cdot \mathcal{G}(\bm{A})^{-1};\\
	\mathsf{D}^2 \mathcal{G} [\bm{A}]^{-1}(\bm{h},\bm{k}) 
	&= 2 \cdot\mathcal{G} (\bm{A})^{-1} \cdot   \mathsf{D} \mathcal{G} [\bm{A}](\bm{k})  \cdot \mathcal{G} (\bm{A})^{-1} \cdot  \mathsf{D} \mathcal{G} [\bm{A}](\bm{k})  \cdot \mathcal{G} (\bm{A})^{-1}\\
	&-\mathcal{G}(\bm{A})^{-1} \cdot \mathsf{D}^2 \mathcal{G}[\bm{A}](\bm{h},\bm{k}) \cdot \mathcal{G}(\bm{A})^{-1}.
	\end{align*}
\end{lemm}
\begin{proof}
	Denote $\mathcal{F}:\bm{A} \mapsto \bm{A}^{-1}$ as the inversion function. Recall the chain rule of the Fr\'{e}chet derivative:
	\begin{align*}
	\mathsf{D} \mathcal{F}\circ \mathcal{G} [\bm{A}](\bm{h}) &= \mathsf{D}\mathcal{F}[\mathcal{G}(\bm{A})] \left(\mathsf{D} \mathcal{G} [\bm{A}](\bm{h})\right);\\
	\quad\mathsf{D}^2 \mathcal{F} \circ \mathcal{G} [\bm{A}](\bm{h},\bm{k}) &= \mathsf{D}^2 \mathcal{F} [\mathcal{G}(\bm{A})] \left(\mathsf{D} \mathcal{G} [\bm{A}](\bm{k}),\mathsf{D} \mathcal{G} [\bm{A}](\bm{k})\right)+ \mathsf{D} \mathcal{F}[ \mathcal{G} (\bm{A})]\left( \mathsf{D}^2 \mathcal{G} [\bm{A}](\bm{h},\bm{k})\right).
	\end{align*}
	By applying the formulas of the Fr\'echet derivative of the inversion function (see e.g.~\cite[Example X.4.2]{Bha97}, and \cite[Exercise 3.27]{HP14}):
	\begin{align*}
	\mathsf{D}[\bm{X}]^{-1} (\bm{Y}) &= - \bm{X}^{-1} \bm{Y} \bm{X}^{-1}\\
	\mathsf{D}^2[\bm{X}]^{-1} (\bm{Y}_1, \bm{Y}_2) &=  \bm{X}^{-1} \bm{Y}_1 \bm{X}^{-1} \bm{Y}_2 \bm{X}^{-1} +
	\bm{X}^{-1} \bm{Y}_2 \bm{X}^{-1} \bm{Y}_1 \bm{X}^{-1}.
	\end{align*}
	concludes the desired results.
	
\end{proof}

\begin{prop}
	[Operator Jensen's Inequality \cite{Dav57, Cho74,HP03,FZ07}] \label{prop:Jensen}
	Let $(\mathrm{\Omega},\mathrm{\Sigma})$ be a measurable space and suppose that $I\subseteq \mathbb{R}$ is an open interval. 
	Assume for every $x \in \Omega$, $\bm{K}(x)$ is a (finite or infinite dimensional) square matrix and satisfies 
	\[ \int_{x\in\Omega} \bm{K}(\mathrm{d}x) \bm{K}(\mathrm{d}x)^\dagger  = \bm{I} \] (identity matrix in $\mathbb{M}^\text{sa}$).
	If $\bm{f}:\mathrm{\Omega}\rightarrow \mathbb{M}^\text{sa}$ is a measurable function for which $\sigma(\bm{f}(x)) \subset I$, for every $x\in \Omega$, then
	\[
	\phi\left( \int_{x\in\Omega} \bm{K}(\mathrm{d}x) \bm{f}(x) \bm{K}(\mathrm{d}x)^\dagger  \right) \preceq  \int_{x\in\Omega} \bm{K}(\mathrm{d}x) \phi\left(\bm{f}(x)\right) \bm{K}(\mathrm{d}x)^\dagger \, \mu(\mathrm{d}x)
	\]
	for every operator convex function $\phi:I\rightarrow \mathbb{R}$.
	Moreover, 
	\[
	\Tr \left[ \phi\left( \int_{x\in\Omega}  \bm{K}(\mathrm{d}x) \bm{f}(x) \bm{K}(\mathrm{d}x)^\dagger \, \mu(\mathrm{d}x) \right) \right] \leq \Tr \left[ \int_{x\in\Omega} \bm{K}(\mathrm{d}x) \phi\left(\bm{f}(x)\right) \bm{K}(\mathrm{d}x)^\dagger \, \mu(\mathrm{d}x) \right]
	\]
	for every convex function $\phi:I\rightarrow \mathbb{R}$.
\end{prop}


\printbibliography

\end{document}